\documentclass[a4paper,reqno,11pt]{amsart}
\usepackage[T1]{fontenc}
\usepackage[utf8]{inputenc}
\usepackage[english]{babel}
\usepackage[dvipsnames,svgnames,table]{xcolor}
\usepackage{amssymb, amsmath, amsthm, graphicx, enumerate, color, mathtools, comment, caption, subcaption, float, lmodern}
\usepackage[foot]{amsaddr}
\usepackage[shortlabels,inline]{enumitem}
\usepackage[unicode=true]{hyperref}
\usepackage[capitalise, noabbrev]{cleveref}
\usepackage{thmtools, thm-restate}
\usepackage{doi}
\DeclarePairedDelimiter\set{\{}{\}}

\newcommand{\q}[1]{``#1''}
\newcommand{\defin}[1]{\emph{\textcolor{ForestGreen}{#1}}}
\DeclareMathOperator\length{len}

\newcommand{\N}{\mathbb{N}}

\newcommand{\Oh}{\mathcal{O}}

\newcommand{\calS}{\mathcal{S}}

\newcommand{\dist}{\mathrm{dist}}

\DeclareMathOperator\parent{parent}

\let\le\leqslant

\let\leq\leqslant
\let\geq\geqslant

\let\subset\subseteq

\let\epsilon\varepsilon

\let\setminus\backslash
\let\preceq\preccurlyeq

\newcommand{\LCA}{\mathrm{LCA}}
\newcommand{\lca}{\mathrm{lca}}

\crefformat{equation}{#2(#1)#3}
\let\eqref\cref
\crefformat{subsection}{Subsection #2#1#3}

\makeatletter
\def\thm@space@setup{
  \thm@preskip=4mm
  \thm@postskip=0mm
}
\makeatother

\usepackage{algorithmicx}
\usepackage[noend]{algpseudocode}
\algdef{SE}[DOWHILE]{Do}{doWhile}{\algorithmicdo}[1]{\algorithmicwhile\ #1}

\usepackage{mdframed}

\mdfdefinestyle{dontsplit}{
  hidealllines=true,
  nobreak=true,
  leftmargin=0pt,
  rightmargin=0pt,
  innerleftmargin=0pt,
  innerrightmargin=0pt,
}

\newmdtheoremenv[style=dontsplit]{theorem}{Theorem}
\newtheorem{lemma}[theorem]{Lemma}
\newmdtheoremenv[style=dontsplit]{obs}[theorem]{Observation}
\newmdtheoremenv[style=dontsplit]{remark}[theorem]{Remark}
\newmdtheoremenv[style=dontsplit]{proposition}[theorem]{Proposition}
\newmdtheoremenv[style=dontsplit]{question}{Question} 
\newmdtheoremenv[style=dontsplit]{corollary}[theorem]{Corollary} 
\newmdtheoremenv[style=dontsplit]{problem}[theorem]{Problem}
\newmdtheoremenv[style=dontsplit]{conjecture}[theorem]{Conjecture}
\newtheorem*{theorem*}{Theorem}
\newtheorem*{corollary*}{Corollary} 
\newtheorem*{problem*}{Problem}
\newtheorem*{conjecture*}{Conjecture}
\newtheorem*{lemma*}{Lemma}
\newtheorem*{obs*}{Observation}
\newtheorem*{remark*}{Remark}
\newtheorem*{proposition*}{Proposition}
\newtheorem*{question*}{Question} 
\newtheorem*{example*}{Example}

\theoremstyle{remark}
\newtheorem*{claim*}{Claim} 
\newmdtheoremenv[style=dontsplit]{claim}[theorem]{Claim}
\crefname{claim}{Claim}{Claims}
\newmdtheoremenv[style=dontsplit]{example}[theorem]{Example}

\newenvironment{proofclaim}[1][]
	{\vspace{-\topsep}\begin{proof}[Proof] }{\end{proof}}

\graphicspath{{figs/}}

\renewenvironment{enumerate}{\begin{enumorig}[label=\textup{(\roman*)}]}{\end{enumorig}}

\makeatletter
\newcommand{\myitem}[1]{%
\item[#1]\protected@edef\@currentlabel{#1}%
}
\makeatother

\hypersetup{ 
    colorlinks,
    linkcolor={RoyalBlue},
    citecolor={RubineRed},
    urlcolor={blue!80!black},
    pdftitle={Sample compression schemes for balls in structurally sparse graphs}
}

\newcommand{\emptyentry}{\bot}
\newcommand{\header}[1]{\textcolor{DarkOrchid}{#1}}

\begin{document} 
\title[Sample compression schemes] 
{Sample compression schemes for balls in structurally sparse graphs}

\author[Bourneuf]{Romain Bourneuf}
\address[R.~Bourneuf]{Univ. Bordeaux, CNRS, Bordeaux INP, LaBRI, UMR 5800, F-33400 Talence, France}
\email{\href{mailto:romain.bourneuf@ens-lyon.fr}{romain.bourneuf@ens-lyon.fr}}

\author[Hodor]{J\k{e}drzej Hodor}
\address[J.~Hodor]{Theoretical Computer Science Department, 
Faculty of Mathematics and Computer Science and  Doctoral School of Exact and Natural Sciences, Jagiellonian University, Krak\'ow, Poland}
\email{\href{mailto:jedrzej.hodor@gmail.com}{jedrzej.hodor@gmail.com}}

\author[Micek]{Piotr Micek}
\address[P.~Micek]{Theoretical Computer Science Department, 
Faculty of Mathematics and Computer Science, Jagiellonian University, Krak\'ow, Poland}
\email{\href{mailto:piotr.micek@uj.edu.pl}{piotr.micek@uj.edu.pl}}

\author[Rambaud]{Clément Rambaud}
\address[C.~Rambaud]{Universit\'e C\^ote d'Azur, CNRS, Inria, I3S, Sophia-Antipolis, France}
\email{\href{mailto:clement.rambaud@normalesup.org}{clement.rambaud@normalesup.org}}

\thanks{J.\ Hodor is supported by a Polish Ministry of Education and Science grant (Perły Nauki; PN/01/0265/2022). 
P.\ Micek is supported by the National Science Center of Poland under grant UMO-2023/05/Y/ST6/00079 within the WEAVE-UNISONO program.
This work was done during a visit of the first and the last authors 
at Jagiellonian University in February 2026.}

\begin{abstract}
Sample compression schemes were defined by Littlestone and Warmuth (1986) as an abstraction of the structure underlying many learning algorithms.
In a sample compression scheme, we are given a large sample of vertices of a fixed hypergraph with labels indicating the containment in some hyperedge.
The task is to compress the sample in such a way that we can retrieve the labels of the original sample.
The size of a sample compression scheme is the amount of information that is kept in the compression.
Every hypergraph with a sample compression scheme of bounded size must have bounded VC-dimension.
Conversely, Moran and Yehudayoff (J.\ ACM, 2016) showed that every hypergraph of bounded VC-dimension admits a sample compression scheme of bounded size.
We study a specific class of hypergraphs emerging from balls in graphs.
The schemes that we construct (contrary to the ones constructed by Moran and Yehudayoff) are \emph{proper}, meaning that we retrieve not only the labeling of the original sample but also a hyperedge (ball) consistent with the original labeling.
First, we prove that for every graph $G$ of treewidth at most $t$, the hypergraph of balls in $G$ has a proper sample compression scheme of size $\Oh(t\log t)$; this is tight up to the logarithmic factor and improves the quadratic (improper) bound that follows from the result of Moran and Yehudayoff.
Second, we prove an analogous result for graphs of cliquewidth at most $t$.
\end{abstract}

\maketitle
%\tableofcontents
%\clearpage

\clearpage

% ==============================================hypergraph=========
\section{Introduction}\label{sec:introduction}
% =======================================================

%Consider a universe $U$ and a family $\mathcal{C}$ of subsets of $U$ that is sometimes called a \emph{concept class}. 
%A \emph{sample} of $U$ is a pair $(X^+,X^-)$ of subsets of $U$ such that there exists $C\in \mathcal{C}$ with $X^+\subseteq C$ and $X^-\cap C = \emptyset$. 
%Consider a universe $U$ of elements and a family $\mathcal{C}$ of subsets of the universe which is sometimes called a \emph{concept class}. 
%Consider a \emph{sample} to be just a list of vertices of $U$ that are labeled with a $+$ or a $-$ in such a way that there is a concept $C$ in $\mathcal{C}$ such that all the elements labeled with $+$ are exactly those elements of the list that lie in $C$. In this case, we say that $C$ \emph{realizes} the sample. 
A \emph{concept class} (or a \emph{hypergraph}) is a family $\mathcal{C}$ of subsets of a universe $U$, called \emph{concepts} (or \emph{hyperedges}).
A \emph{sample} of $\mathcal{C}$ is a list of elements of $U$ that are labeled with a $+$ or a $-$ in such a way that there is a concept $C$ in $\mathcal{C}$ such that the elements labeled with $+$ are exactly those elements of the list that lie in $C$. In this case, we say that $C$ \emph{realizes} the sample. 
One should think of a sample as a list of positive and negative examples corresponding to a certain concept $C$ in $\mathcal{C}$.
A \emph{sample compression scheme} takes an arbitrarily long sample and compresses it into a short sample in a way that allows reconstructing a concept $C'$ which realizes the original input sample.
%A sample compression scheme over some universe, given a labeled \emph{sample} of elements wants to find a small subsample that retains all the information about the labels of the initial sample.
%For sample compression schemes to exist, we need some restrictions on the structure of labels.
%We consider binary labels and we assume that the labels of a sample correspond to a hyperedge of a given hypergraph.
A mathematical framework for sample compression scheme was introduced in 1986 by Littlestone and Warmuth~\cite{LW86} and has since then become a central topic in learning theory.
%There are some variants of the definition of the scheme and we choose to follow the one allowing a subsample to be ordered. See the end of this section for a discussion comparing the variants.
We introduce it formally below.

A \defin{hypergraph} is a pair $(V,E)$ where $V$ is a finite set,
and $E$ is a set of subsets of $V$. 
In a hypergraph $H=(V,E)$, the elements of $V$ are called \defin{vertices} of $H$ and the elements of $E$ are called \defin{hyperedges} of $H$; we also write $\defin{V(H)} = V$ and $\defin{E(H)} = E$.

Let $H$ be a hypergraph.
A \defin{sample} of $H$ is a pair $(X^+,X^-)$
of subsets of $V(H)$ such that
there is a hyperedge $e \in E(H)$ with $X^+ \subseteq e$ and $X^- \cap e = \emptyset$.
%For every such hyperedge $e$, we say that $e$ \defin{realizes} $(X^+,X^-)$.
A set $S \subset V(H)$ \defin{realizes} a sample $(X^+, X^-)$ if $X^+ \subseteq S$ and $S \cap X^- = \emptyset$.
We denote the set of all samples of $H$ by \defin{$\calS(H)$}.
A sample $(Y^+,Y^-) \in \calS(H)$ is a \defin{subsample} of $(X^+,X^-)\in \calS(H)$ if $Y^+ \subset X^+$ and $Y^- \subset X^-$.
The \defin{size} of a sample $(X^+,X^-)$ of $H$ is $|X^+ \cup X^-|$.
A \defin{sample compression scheme} of $H$ is a pair of functions $(\kappa,\rho)$ where\footnote{With a slight abuse of notation, we omit double parentheses when applying functions to pairs or tuples. We denote by \defin{$\{0,1\}^*$} the set of all finite sequences with $0$-$1$ entries (bitstrings).}
\[
    \kappa\colon \calS(H) \rightarrow \calS(H) \times \{0,1\}^* \text{ and } \rho\colon \calS(H) \times \{0,1\}^* \rightarrow 2^{V(H)}
\]
such that for every $(X^+,X^-) \in \calS(H)$, we have
\begin{enumerate}
    \item $\kappa_1(X^+,X^-)$ is a subsample of $(X^+,X^-)$ and
    \item $\rho(\kappa(X^+,X^-))$ realizes $(X^+,X^-)$,
\end{enumerate}
where $(\kappa_1(X^+,X^-),\kappa_2(X^+,X^-)) = \kappa(X^+,X^-)$.  
For each $i\in\set{1,2}$, let $k_i$ be the maximum over all $(X^+,X^-)\in\calS(H)$ of the size of $\kappa_i(X^+,X^-)$. 
The \defin{size} of the sample compression scheme $(\kappa, \rho)$ is $k_1+k_2$.
We refer to $\kappa$ as the \defin{compressor} and to $\rho$ as the \defin{reconstructor}.
A sample compression scheme as above is \defin{proper} if all values of the reconstructor $\rho$ are in $E(H)$.

% It is easy to check that for a hypergraph where every subset of its vertices is a hyperedge no nontrivial sample compression scheme exists.
% On the other hand, a hypergraph with exactly one hyperedge admits a sample compression scheme of size $0$.
% Thus, it is natural to ask what structural restriction must be imposed on hypergraphs so that they admit a sample compression scheme of small size.

There is a simple example of hypergraphs that do not admit sample compression schemes of small size: let $H_n = ([n], 2^{[n]})$.\footnote{For each nonnegative integer $n$, we denote by \defin{$[n]$} the set $\{1,\dots,n\}$.} 
%\piotr{Maybe we can write next sentences better.} 
If $(\kappa, \rho)$ is a sample compression scheme of $H_n$, then the restriction of $\kappa$ to samples of size $n$ is injective, and a simple counting argument shows that the size of $(\kappa, \rho)$ is unbounded.
More generally, any hypergraph containing $H_n$ can not have a sample compression scheme of small size. This motivates the study of hypergraphs of bounded VC-dimension, that we now define.

A subset $X$ of the vertices of a hypergraph $H$ is \defin{shattered} by $H$ if every bipartition $(X^+, X^-)$ of $X$ is a sample of $H$.
The maximum size of a shattered set in a hypergraph is known as its Vapnik--Chervonenkis dimension~\cite{VC71}, or simply \defin{VC-dimension}.
Observe that the VC-dimension of a hypergraph $H$ is the largest integer $n$ such that $H$ contains a subhypergraph isomorphic to $H_n$.

Littlestone and Warmuth in their foundational paper~\cite{LW86} proved that a hypergraph which admits a sample compression scheme of bounded size has bounded VC-dimension (or equivalently is PAC-learnable, see~\Cref{sec:VC-dim}).
%We give more context on PAC-learnability and VC-dimension in~\Cref{sec:VC-dim}.
The converse implication remained a tantalizing problem for the next 30 years, until it was finally settled in a remarkable paper by Moran and Yehudayoff \cite{Moran2016}.
They proved that every hypergraph of VC-dimension at most $d$ admits a sample compression scheme of size $2^{\Oh(d)}$. 
More precisely, they proved that
a hypergraph of VC-dimension at most $d$ admits a sample compression scheme of size $\Oh(k\log k)$ where $k = d \cdot d^*$ and $d^*$ is the VC-dimension of the dual hypergraph.\footnote{For a hypergraph $H$, the \defin{dual hypergraph} of $H$ has vertex set $E(H)$ and edge set $\set{\set{e \in E(H)\mid v\in e}\mid v\in V(H)}$.} 
As it is well-known that $d^* < 2^{d+1}$, see e.g.\ \cite{A83}, 
this gives a $2^{\Oh(d)}$ upper bound on the size of the scheme.

The \emph{sample compression conjecture}, due to Floyd and Warmuth \cite{FW95}, states that every hypergraph of VC-dimension $d$ admits a sample compression scheme of size linear in $d$. Warmuth \cite{W03} offered a \$600 reward for a solution. The conjecture remains wide open.
Note that a linear bound would be indeed optimal.
A direct counting argument due to Floyd and Warmuth \cite{FW95} shows that for every $d$, there exists a hypergraph of VC-dimension $d$ for which every sample compression scheme has size at least $d$, and that for every hypergraph of VC-dimension $d$, every sample compression scheme must have size at least $d/5$.
Moreover, Pálvölgyi and Tardos \cite{PT20} constructed a hypergraph of VC-dimension 2 which does not admit a sample compression scheme of size $2$.

Due to little progress on the sample compression conjecture, more attention has been given to special cases of this conjecture. 
In this paper, we prove almost tight bounds for hypergraphs of balls in structurally sparse graphs.

Let $G$ be a graph. For $u,v \in V(G)$, we denote by \defin{$\dist_G(u,v)$} the \defin{distance} between $u$ and $v$ in $G$, i.e.\ the minimum number of edges in a $u$-$v$ path in $G$.
For $c \in V(G)$ and an integer $r$, the \defin{ball} in $G$ of \defin{radius} $r$ \defin{centered} in $c$ is the set $\defin{\text{$B_G(c,r)$}} = \{v \in V(G) \mid \dist_G(c,v) \leq r\}$.
%For $c \in V(G)$, $B_G(c,1) \setminus\{c\}$ is the set of \defin{neighbors} of $c$ in $G$ and is denoted by \defin{$N_G(c)$}.
The \defin{hypergraph of balls} in $G$ is the hypergraph $H$ with $V(H) = V(G)$ and where $E(H)$ is the collection of all the balls in $G$, i.e.\ $B_G(c,r)$ for each $c \in V(G)$ and for each integer $r$.

%\piotr{remove}
%When a graph $G$ is well-structured, then the hypergraph of balls in $G$ has bounded VC-dimension, and it is interesting to investigate whether this hypergraph admits a sample compression scheme of size linear (or almost linear) in its VC-dimension.
%\rom{Add that the general case is as hard as the sample compression conjecture.}

%Chalopin, Chepoi, Inerney, Ratel, and Vaxès~\cite{Chalopin2023} studied sample compression schemes for hypergraphs of balls of graphs.
%They gave proper sample compression schemes of constant size for trees, cycles, interval graphs, and cacti.
%\piotr{end remove}

Bousquet and Thomassé \cite{BT15} proved that the hypergraph of balls in a $K_t$-minor-free graph has VC-dimension at most $t-1$. 
Building on their result, Beaudou, Dankelmann, Foucaud, Henning, Mary, and Parreau~\cite[Theorem~18]{Beaudou2018} proved that the dual of the hypergraph of balls in a $K_t$-minor-free graph also has VC-dimension at most $t-1$. 
Therefore, the general result of Moran and Yehudayoff \cite{Moran2016} applied to the hypergraph of balls in a $K_t$-minor-free graph (or a graph of treewidth at most $t$) gives a sample compression scheme of size $\mathcal{O}(t^2\log t)$. 
%For comparison, note that~\Cref{thm:tw-first} gives a proper scheme of smaller size.
To the best of our knowledge there is no \emph{proper} sample compression scheme of size bounded by a function of $t$ for the hypergraph of balls in a graph of treewidth at most $t$.
The following theorem is one of the main contributions of the paper.

\begin{theorem}\label{thm:tw-first}
For every graph $G$ of treewidth at most $t$, the hypergraph of balls in $G$ has a proper sample compression scheme of size $\Oh(t\log t)$.
\end{theorem}

We remark that no sublinear bound in~\Cref{thm:tw-first} is possible. 
For a positive integer $t$, 
consider the bipartite graph $G_t$ on the vertex set $[t] \cup 2^{[t]}$ such that $[t]$ and $2^{[t]}$ are independent sets and for all $A \in 2^{[t]}$, $A$ is adjacent in $G_t$ to exactly the elements of $A$ in $[t]$.
It is easy to verify that $G_t$ has treewidth $t$ and the hypergraph of balls in $G_t$ has VC-dimension at least $t$.
By the result of Floyd and Warmuth \cite{FW95}, every sample compression scheme of the hypergraph of balls in $G_t$ has size at least $t \slash 5$.
This shows that~\Cref{thm:tw-first} is optimal up to the logarithmic factor.

Note that~\Cref{thm:tw-first} implies that every chordal graph $G$ admits a proper sample compression scheme of size $\mathcal{O}(\omega(G) \log \omega(G))$, 
where $\omega(G)$ stands for the clique number of $G$.
This resolves a problem stated by Chalopin, Chepoi, Inerney, Ratel, and Vaxès~\cite{Chalopin2023}.

While treewidth measures the \emph{tree-likeness} of sparse graphs, cliquewidth is a graph parameter that measures tree-likeness of also dense graphs.
%Roughly, the cliquewidth of a graph is the number of labels necessary to build the graph using some simple operations in a tree-line manner \clement{sentence to improve or remove}.
Roughly, the cliquewidth of a graph measures how a graph can be constructed by two
operations: disjoint unions, and addition of a ``simple'' cut.
Cliquewidth was first described explicitly by Courcelle and Olariu~\cite{CO00}, though the idea dates back to the work on graph grammars in the 1990s.
This parameter is a dense analogue of treewidth, in the sense that every graph of treewidth at most $t$ has cliquewidth at most $3 \cdot 2^{t-1}$, see~\cite{CR05}, and every graph of cliquewidth $t$ which does not contain the complete bipartite graph $K_{s, s}$ as a subgraph has treewidth at most $3t(s-1)-1$, see~\cite{GW00}.
Cliquewidth is functionally equivalent to \emph{rankwidth}, which is another width parameter on graphs. More precisely, every graph of cliquewidth $t$ has rankwidth at least $t$ and at most $2^{t+1} - 1$, see~\cite{OS06}.
Bousquet and Thomassé~\cite{BT15} proved that the hypergraph of balls in a graph of rankwidth at most $t$ has VC-dimension at most $3 \cdot 2^{t+1} + 2$.
Thus, the hypergraphs of balls in graphs of bounded cliquewidth have bounded VC-dimension.
Our second main result is an almost optimal proper sample compression scheme for hypergraphs of balls in graphs of bounded cliquewidth.

\begin{theorem}\label{thm:cw-first}
For every graph $G$ of cliquewidth at most $t$, the hypergraph of balls in $G$ has a proper sample compression scheme of size $\Oh(t\log t)$.
\end{theorem}

As mentioned before, Floyd and Warmuth~\cite{FW95} proved that every sample compression scheme of a hypergraph of VC-dimension $d$ has size at least $d/5$.
Therefore,~\Cref{thm:cw-first} implies that the hypergraphs of balls in graphs of cliquewidth at most $t$ have VC-dimension $\Oh(t \log t)$. Actually, adapting the proof of Bousquet and Thomassé that hypergraphs of balls in graphs of bounded rankwidth have bounded VC-dimension, we prove that hypergraphs of balls in graphs of cliquewidth at most $t$ have VC-dimension $\Oh(t)$, see \cref{prop:bound-vcdim}.\footnote{Since the proof follows step by step the proof of Bousquet and Thomassé~\cite{BT15}, we present it in the appendix.}
Again, it is easy to check that the cliquewidth of $G_t$ is at most $t+1$, hence, the bound in~\Cref{thm:cw-first} (respectively \Cref{prop:bound-vcdim}) is tight up to a logarithmic (respectively constant) factor.

Twin-width, introduced by Bonnet, Kim, Thomass\'{e}, and Watrigant~\cite{BKTW21}, is a width parameter that has attracted a lot of attention in the last years.
This parameter generalizes cliquewidth, in the sense that graphs of bounded cliquewidth have bounded twin-width.
It is natural to wonder whether \cref{thm:cw-first} can be extended to graphs of bounded twin-width. 
This is not possible since the hypergraphs of balls in graphs of bounded twin-width can have arbitrarily large VC-dimension.
This follows from the fact that any $(\geq 2\log n)$-subdivision of an $n$-vertex graph has twin-width at most $4$ \cite{BBD22,BGKTW21}.
Then, starting from an $n$-vertex graph $G$ whose hypergraph of balls has arbitrary VC-dimension $t$ (e.g.\ $G_t$ defined before), 
the $(2 \log n)$-subdivision of $G$ has twin-width at most $4$, but its hypergraph of balls has VC-dimension at least $t$.

Although we tried, we do not see how to remove the logarithmic factors from the statements of~\Cref{thm:tw-first,thm:cw-first}. 
We are only able to devise a smaller proper sample compression under a stronger structural assumption on the graphs we consider, namely bounded vertex cover number. 
% For every graph $G$, treewidth, pathwidth, treedepth, and vertex cover number of $G$ satisfy 
% \[
% \tw(G)+1 \leq \pw(G)+1 \leq \td(G) \leq \vc(G).
% \]
% We continue with the definition of the vertex cover number. 
A \defin{vertex cover} in a graph $G$ is a subset $R$ of vertices of $G$ such that every edge of $G$ is incident to a vertex of $R$. 
%The \defin{vertex cover number} of $G$, denoted by \defin{$\vc(G)$}, is the minimum size of a vertex cover in $G$. 

\begin{restatable}{theorem}{thmvc}\label{thm:vc}
For every graph $G$ admitting a vertex cover of size at most $t$, the hypergraph of balls in $G$ has a proper sample compression scheme of size $t+4$.
\end{restatable}

% It would be interesting to have a similar statement for graphs of treedepth at most $t$, see Open Problems Section.

% Littlestone \cite{L88} defined a complexity measure for hypergraphs, which is now known as \defin{Littlestone dimension}. Bounded Littlestone dimension characterizes online learnability \cite{BPS09} and private PAC learnability \cite{ABLMM22}. Every hypergraph of VC-dimension $d$ has Littlestone dimension at least $d$, and there exist hypergraphs of VC-dimension $2$ with arbitrarily large Littlestone dimension. 
% See more details in~\Cref{sec:littlestone}.
% Due to the strong connections between Littlestone dimension and learning, it is natural to expect that hypergraphs of bounded Littlestone dimension admit small sample compression schemes. We prove that it is indeed the case.

% \begin{restatable}{theorem}{thmld}\label{thm:ld}
%     Every hypergraph of Littlestone dimension $d$ admits a sample compression scheme of size $d$.
% \end{restatable}

% This bound is again tight up to constant factors.

Finally, we remark that our proof of \cref{thm:tw-first} also works in metric spaces of bounded treewidth, i.e.\ if the edges of the graph $G$ have nonnegative weights and the length of a path is the sum of the weights of the edges along this path. This is not true for \cref{thm:cw-first}. An intuitive explanation for this phenomenon is that the edge weights can be emulated by subdividing the edges of the initial graph, and subdividing edges can not increase the treewidth, but can increase the cliquewidth.

\subsection{Balls of bounded radius}\label{sec:intro:bd-radius}

The hypergraph of balls in a planar graph has VC-dimension at most $4$ as proved by Bousquet and Thomassé~\cite{BT15}. 
Thus, the result of Moran and Yehudayoff~\cite{Moran2016} implies that hypergraphs of balls in planar graphs have sample compression schemes of bounded size.
Chalopin et al.~\cite{Chalopin2023} asked if proper sample compression schemes can be constructed in this setting and they provided some partial answers.

Let $I$ be a set of integers and let $G$ be a graph.
The \defin{hypergraph of balls of radii in $I$} in $G$ is the hypergraph $H$ with $V(H) = V(G)$ and $E(H)$ given by the collection of all the balls in $G$ of radius in $I$, i.e.\ $B_G(c,r)$ for each $c \in V(G)$ and for each integer $r \in I$.
In the specific case of $I = \{1\}$, we say that $H$ is the \defin{hypergraph of closed neighborhoods} in $G$.

Chalopin et al.~\cite{Chalopin2023} constructed a proper sample compression scheme of constant size for the hypergraph of closed neighborhoods in any planar graph.
We extend this result to any radius.

\begin{restatable}{theorem}{thmplan}\label{th:planar-bd-radius}
    For every planar graph $G$ and for every positive integer $r$, the hypergraph of balls of radius at most $r$ in $G$ has a proper sample compression scheme of size $\Oh(r\log r)$.
\end{restatable}

\Cref{th:planar-bd-radius} follows from a more general result on graphs of bounded \emph{local treewidth}, see~\Cref{sec:local-treewidth}.

We also give a simple proof that the hypergraph of closed neighborhoods in a graph of degeneracy at most $t$ has a sample compression scheme of size $\Oh(t)$.
The result of Moran and Yehudayoff~\cite{Moran2016} implies a sample compression scheme of size $\Oh(t^2\log t)$, and the linear bound is best possible.
Also note that this result is optimal in the sense that closed neighborhoods can not be replaced by balls of radius $2$.
See more details in~\Cref{sec:bd-degeneracy}.

\begin{restatable}{theorem}{thmdeg}\label{thm:degeneracy}
    For every positive integer $t$ and every graph $G$ of degeneracy at most $t$, the hypergraph of closed neighborhoods in $G$ has a sample compression scheme of size $t + \lceil \log(t+1)\rceil + 1$. 
\end{restatable}

\subsection{Learning and VC-dimension}\label{sec:VC-dim}

There are many ways to define what a learning algorithm should do. The model of \emph{probably approximately correct} (\emph{PAC}) learning, introduced by Valiant \cite{V84}, is one of the most studied frameworks.
% Given a concept class $\mathcal{C}$, a \defin{sample} of $\mathcal{C}$ is a pair $(S^+, S^-)$ such that there exists a concept $C \in \mathcal{C}$ such that $S^+ = (S^+ \cup S^-) \cap C$. We denote by $\mathcal{S}(\mathcal{C})$ the set of samples of $\mathcal{C}$.
Informally, an algorithm $\mathcal{A}$ learns a hypergraph $H$ if it satisfies the following property: for every hyperedge $e$ of $H$, for every probability distribution $\mu$ on the universe $V(H)$, given a sample $(S \cap e, S \setminus e)$ where $S$ consists of $s$ elements sampled independently with repetition from $\mu$, $\mathcal{A}$ outputs a hyperedge $e'$ of $H$ such that $e$ and $e'$ are very close for the distribution $\mu$.
Formally, an algorithm $\mathcal{A} \colon \mathcal{S}(H) \to V(H)$ is a \defin{PAC learning algorithm} for a hypergraph $H$ if for all constants $\varepsilon, \delta > 0$, there exists a constant $s = s(\varepsilon, \delta)$ such that for any edge $e$ of $H$, for every probability distribution $\mu$ on $V(H)$, we have \[
\mathbb{P}_{S \gets \mu^s}[\mu(\mathcal{A}(S \cap e, S \setminus e) \Delta e) \geq \varepsilon] \leq \delta.
\]

%A subset $X$ of the vertices of a hypergraph $H$ is \defin{shattered} by $H$ if every bipartition $(X^+, X^-)$ of $X$ is a sample of $H$.
If a hypergraph $H$ contains arbitrarily large shattered sets then $H$ is not PAC-learnable: for any fixed $\varepsilon, \delta, s > 0$, let $X$ be a large enough shattered set (as a function of $\varepsilon,$ $\delta$, and $s$), and let $\mu$ be the uniform distribution on $X$.
% The size of the largest shattered set in a concept class (or $\infty$ if there exist arbitrarily large shattered sets) is a measure of complexity of the concept class which was studied by Vapnik and Chervonenkis \cite{VC71}, and is now known as \defin{VC-dimension}.
It follows that hypergraphs with infinite VC-dimension are not PAC-learnable. 
Building on the work of Vapnik and Chervonenkis, Blumer, Ehrenfeucht, Haussler and Warmuth \cite{BEHW89} proved that hypergraphs of bounded VC-dimension are PAC-learnable, so bounded VC-dimension characterizes PAC learnability.

\subsection{Related results}

%\piotr{We need to put a sentence or two here explaining unlabeled schemes.}
%\rom{I think this section is mostly to give pointers, in my opinion we should not bother with defining stuff here.}

%\jedrzej{I understand very little from the next paragraph. Maybe it's okay, maybe we should rethink how it looks like.}\rom{In my opinion, this section is mostly to give pointers, so I think it's okay as it is, but I would not fight to keep it as is.}
%\piotr{The meaning of \emph{unlabeled} is not clear. It probably means that there is no extra bitstring.}
%\jedrzej{I also don't know what is \q{maximum concept class} or \q{extremal concept class}.} 
%\piotr{Adding anything increases VC-dimension.}
%\jedrzej{Then why not maximal?} 
%\piotr{Then maybe the largest one with this property.}
%\piotr{Dudley class is also mysterious.}
%\jedrzej{Yes, so basically this is my concern.}
%\rom{How about adding a sentence like: "We now review some related results. All the relevant definitions can be found in the corresponding articles."}
%\jedrzej{Sounds good to me.} \piotr{Yes.}

Next, we review some related results. All the relevant definitions can be found in the corresponding articles.
In their paper, Floyd and Warmuth \cite{FW95} showed that maximum concept classes of VC-dimension $d$ admit proper sample compression schemes of size $d$. This was later extended by Moran and Warmuth \cite{MW16} to extremal concept classes of VC-dimension $d$.
Chalopin, Chepoi, Moran and Warmuth \cite{CCMW22} exhibited \emph{unlabeled} sample compression schemes of size $d$ for maximum concept classes of VC-dimension $d$.
Floyd and Warmuth \cite{FW95} also conjectured that every concept class of VC-dimension $d$ is contained in a maximum concept class of VC-dimension $\mathcal{O}(d)$. If true, this embedding conjecture would imply the sample compression conjecture by combining it with the result of Floyd and Warmuth \cite{FW95}. 
Special cases of this conjecture were confirmed. For instance, Ben-David and Litman \cite{BL98} proved that every Dudley class of VC-dimension $d$ can be embedded into a maximum class of VC-dimension $d$.
In particular, this implies that there is a sample compression scheme of size $d$ for balls in $\mathbb{R}^{d-1}$.
The embedding conjecture was recently refuted by Chase, Chornomaz, Hanneke, Moran, and Yehudayoff \cite{CCHMY24}, who proved that there exist concept classes of VC-dimension $d$ such that every extremal class that contains them has VC-dimension at least exponential in $d$.
Chepoi, Knauer, and Philibert \cite{CKP24} proved that the topes of a complex of oriented matroids of VC-dimension $d$ admit a proper sample compression scheme of size $d$, extending their previous result \cite{CKP22}, which relied on embedding into extremal concept classes. This generalizes both \cite{MW16} and \cite{BL98}.
Marc \cite{M24} proved that oriented matroids of VC-dimension $d$ admit proper unlabeled sample compression schemes of size $d$.

\subsection{Overview of the paper}
In~\cref{sec:preliminaries}, we set basic notation, we state technical variants of~\Cref{thm:tw-first,thm:cw-first} 
% \rom{Don't we want to say~\Cref{thm:tw-first,thm:cw-first} here? Same in the next sentences.} \jedrzej{Yes, thanks, I mixed the references.}
that we actually prove, and we recall the definitions of the width parameters that we work with.
In~\Cref{sec:overview}, we give high-level overviews of the proofs of~\Cref{thm:tw,thm:cw}.
\Cref{sec:subtrees} contains the common part of the proofs of the main results.
In~\Cref{sec:treewidth} we prove~\Cref{thm:tw-first}, in~\Cref{sec:cliquewidth}, we prove~\Cref{thm:cw-first}, and in~\cref{sec:vc}, we prove~\Cref{thm:vc}.
In~\Cref{sec:bd-radius}, we elaborate on the material discussed in~\Cref{sec:intro:bd-radius}, in particular, we prove~\Cref{thm:degeneracy,th:planar-bd-radius}.
Finally, in~\Cref{sec:open}, we discuss some open problems.

% =======================================================
\section{Preliminaries}\label{sec:preliminaries}
% =======================================================

For convenience, we set $\min \emptyset = \infty$ and $\max \emptyset = -1$. 

Let $G$ be a graph. 
For a vertex $v$ of $G$, the set of all neighbors of $v$ in $G$ is denoted by \defin{$N_G(v)$}. 
Observe that $N_G(v) = B_G(v,1) \setminus \{v\}$.
For a subset $X$ of vertices of $G$, we write $\defin{\text{$N_G(X)$}} = \bigcup_{x\in X} N_G(x) \setminus X$ and $\defin{\text{$B_G(X,r)$}}=\bigcup_{x\in X}B_G(x,r)$.
A \defin{separation} of $G$ is a pair $(C,D)$ of subsets of $V(G)$ such that $V(G) = C \cup D$ and there is no edge in $G$ with one endpoint in $C \setminus D$ and the other in $D \setminus C$.
The \defin{length} of a path $P$ is the number of edges of $P$, and we denote it by \defin{$\length(P)$}. 
A \defin{$C$-$D$ path} in $G$ is a path in $G$ from a vertex in $C$ to a vertex in $D$ with no internal vertices in $C\cup D$. 
When one of these sets is a single vertex, e.g.\ $C=\set{c}$, we often write a $c$-$D$ path in $G$ instead of a $\set{c}$-$D$ path in $G$.
We denote by \defin{$\dist_G(C,D)$} the minimum length of a $C$-$D$ path in $G$.
We use the same convention for dealing with singletons.
%For $u,v \in V(G)$, we denote by \defin{$\dist_G(u,v)$} the \defin{distance} between $u$ and $v$ in $G$, i.e.\ the length of a shortest $u$-$v$ path in $G$.
%For $c \in V(G)$ and an integer $r$, a \defin{ball} in $G$ of \defin{radius} $r$ \defin{centered} in $c$ is the set $\defin{\text{$B_G(c,r)$}} = \{v \in V(G) \mid \dist_G(c,v) \leq r\}$.
%Note that if $r=1$, we have $B_G(c,1)=N_G(c)\cup \set{c}$; if $r=0$, we have $B_G(c,0)=\set{c}$; and if $r<0$, we have $B_G(c,r)=\emptyset$. 
%For a subset $X$ of vertices of $G$, we define $\defin{\text{$B_G(X,r)$}}=\bigcup_{x\in X}B_G(x,r)$.
%For $c \in V(G)$, $B_G(c,1) \setminus\{c\}$ is the set of \defin{neighbors} of $c$ in $G$ and is denoted by \defin{$N_G(c)$}.

%The \defin{hypergraph of balls} of $G$ is the hypergraph $H$ with $V(H) = V(G)$ and $E(H)$ is the collection of all the balls in $G$, i.e.\ $B_G(c,r)$ for each $c \in V(G)$ and for each integer $r$.

A \defin{rooted tree} is a tree with a distinguished node, called the \defin{root} of $T$.
Let $T$ be a rooted tree. 
A \defin{leaf} in $T$ is a node of degree $1$ which is not the root.
For every node $x$ in $T$ other than the root, we denote by \defin{$\parent(T,x)$} the \defin{parent} of $x$ in $T$.
Then, $x$ is a \defin{child} of $\parent(T, x)$.
For two nodes $x$, $y$ in $T$, $x$ is a \defin{descendant} of $y$ and $y$ is an \defin{ancestor} of $x$ in $T$ if $y$ lies on the path from the root to $x$ in $T$.
In this case, $x$ and $y$ are \defin{comparable}.
Given an edge $xy$ of $T$, we denote by \defin{$T_{x \mid y}$} the component of $x$ in $T \setminus xy$.
Let $x$ and $y$ be two (not necessarily distinct) nodes of $T$.
The \defin{lowest common ancestor} of $x$ and $y$ in $T$, denoted by \defin{$\lca(T,x,y)$}, is the furthest node from the root that is an ancestor of $x$ and $y$.
Let $Z \subset V(T)$.
The \defin{lowest common ancestor closure} of $Z$ in $T$ is the set $\defin{\text{$\LCA(T,Z)$}}=\{\lca(T,x,y)\mid x,y\in Z\}$.  
Observe that $\LCA(T,\LCA(T,Z)) = \LCA(T,Z)$.

A rooted tree $T$ is \defin{binary} if every node $v$ of $T$ has at most two children, among which at most one it the \defin{left child} of $v$ and at most one is the \defin{right child} of $v$.

\subsection{Array sample compression schemes}

In the proofs of~\Cref{thm:tw-first,thm:cw-first}, as well as in the work of Moran and Yehudayoff \cite{Moran2016} and of Chalopin, Chepoi, Mc Inerney, Ratel, and Vaxès~\cite{Chalopin2023}, the additional bitstring in the sample compression scheme is used to encode a certain ordering of the subsample.
Therefore, following Ben-David and Litman~\cite{BL98}, we define another variant of sample compression schemes, which is almost equivalent to the original one.

Let $H$ be a hypergraph.
Let \defin{$\emptyentry$} be a fresh symbol not in $V(H)$, which is used to denote an empty entry in the array.
An \defin{array sample compression scheme} of $H$ of \defin{size} $k$ is a pair of functions $(\kappa,\rho)$ where
\[
    \kappa\colon \calS(H) \rightarrow (V(H) \cup \{\emptyentry\})^k \text{ and } \rho\colon (V(H) \cup \{\emptyentry\})^k \rightarrow 2^{V(H)}
\]
such that for every $(X^+,X^-) \in \calS(H)$, we have
\begin{enumerate}
    \item $\kappa(X^+,X^-) \in (X^+ \cup X^- \cup \{\emptyentry\})^k$ and
    \item $\rho(\kappa(X^+,X^-))$ realizes $(X^+,X^-)$.
\end{enumerate}
Again, such an array sample compression scheme $(\kappa, \rho)$ is \defin{proper} 
if all values of the reconstructor $\rho$ are in $E(H)$.

%\piotr{Maybe we should explain more this equivalence?}

The main difference with (non-array) sample compression schemes is that the values returned by the compressor are ordered sequences, possibly with empty entries.
It is easy to check that usual (non-array) and array sample compression schemes are equivalent up to a logarithmic factor.
More precisely, every (proper) array sample compression scheme of size $k$ can be translated to a (proper) sample compression scheme of size $\mathcal{O}(k\log k)$ by recording a permutation of the elements of the subsample.
In this paper we prove the following statements.

\begin{restatable}{theorem}{thmtw}\label{thm:tw}
For every positive integer $t$ and every graph $G$ of treewidth at most $t$, 
the hypergraph of balls in $G$ has a proper array sample compression scheme of size $4t+7$.
\end{restatable}

\begin{restatable}{theorem}{thmcw}\label{thm:cw}
For every positive integer $t$ and every graph $G$ of cliquewidth at most $t$, 
the hypergraph of balls in $G$ has a proper array sample compression scheme of size $4t+3$.
\end{restatable}

As discussed, \Cref{thm:tw,thm:cw} imply \Cref{thm:tw-first,thm:cw-first}, respectively.
We note that if we assume treedepth at most $t$ instead of treewidth in~\Cref{thm:tw}, with our proof techniques, one can almost directly obtain a multiplicative factor of $2$ in the bound.

\subsection{Width parameters}

A \defin{tree decomposition} of $G$ is a pair $\big(T,(W_x \mid x\in V(T))\big)$, 
where $T$ is a tree and the sets $W_x$ for $x \in V(T)$ are subsets of $V(G)$ 
called \defin{bags} satisfying:
\begin{enumerate}
\item for each edge $uv\in E(G)$ there is a bag containing both $u$ and $v$, and
\item for each vertex $v\in V(G)$ the set of nodes $x\in V(T)$ such that 
$v\in W_x$  induces a non-empty subtree of $T$.
\end{enumerate}
An \defin{adhesion} in this tree decomposition is a set of the form $W_x \cap W_y$
for some edge $xy \in E(T)$.
The \defin{width} of this tree decomposition is $\max\{|W_x|-1 \mid x\in V(T)\}$.
The \defin{treewidth} of $G$ is the minimum width of a tree decomposition of $G$.
It is well-known that every graph of treewidth $t$ has a tree decomposition $\big(T,(W_x \mid x\in V(T))\big)$ of width $t$ where $T$ is a binary tree.

Instead of cliquewidth, we will consider in the proof of \Cref{thm:cw} another graph parameter, called NLC-width\footnote{NLC stands here for \emph{node label controlled}.}, which was
introduced by Wanke~\cite{Wanke94}.
This will be enough for our purposes thanks to the following property:
for every integer $k$,
a graph of cliquewidth at most $k$ has NLC-width at most $k$,
and a graph of NLC-width at most $k$ has cliquewidth at most $2k$~\cite{J01}.

An \defin{NLC-decomposition} of a graph $G$ is a tuple $(T,Q,\alpha,\beta,R)$ where 
\begin{itemize}
    \item $T$ is a rooted binary tree with leaf set $V(G)$,
    \item $Q$ is a finite set of \defin{labels},
    \item $\alpha\colon V(G) \rightarrow Q$ assigns \defin{initial labels},
    \item $\beta\colon \binom{V(T)}{2} \rightarrow (Q \rightarrow Q)$ is a \defin{relabeling function},
    \item $R\colon V(T) \rightarrow Q \times Q$ assigns \defin{status relations},
\end{itemize}
such that 
\begin{enumorig}[label=(nlc\arabic*)]
    \item \label{item:nlc1} for all distinct $u,v \in V(T)$ with $u$ being an ancestor of $v$ if $u_0\cdots u_n$ is the path between $u$ and $v$ in $T$ (i.e.\ $u_0 = u$ and $u_n = v$), then $\beta(\{u_0,u_1\}) \circ \dots \circ \beta(\{u_{n-1},u_n\}) = \beta(\{u,v\})$; and 
    \item \label{item:nlc2} for all distinct $x,y \in V(G)$, if $u$ is the lowest common ancestor of $x$ and $y$ in $T$, 
    $x$ is a descendant of the left child of $u$, and $y$ is a descendant of the right child of $u$, then
    \[
        \textrm{$xy\in E(G)$ if and only if $\big(\beta(\{u,x\})(\alpha(x)),\beta(\{u,y\})(\alpha(y))\big)\in R(u)$.}
    \]
\end{enumorig}
Intuitively, \ref{item:nlc2} says that $R$ determines the adjacency of $x$ and $y$ in $G$ according to the labels obtained by relabeling with $\beta$ the initial labels $\alpha(x)$ and $\alpha(y)$.
%Intuitively, \ref{item:nlc2} says that $R(u)$ describes the set of all the edges between
%left- and right-descendants of $u$ in $T$, using only the labels of these vertices at $u$.
% The \defin{NLC-width} of a graph $G$ is the least $q$ such that $G$ admits an NLC-decomposition that uses a label set of size $q$.

The \defin{NLC-width} of $G$ is then the minimum integer $k$ such that
there is an NLC-decomposition $(T,Q,\alpha, \beta, R)$ of $G$ with $|Q| \leq k$.

Unlike cliquewidth, NLC-width can not increase by complementation:
if $(T,Q,\alpha,\beta,R)$ is an NLC-decomposition of $G$,
then $\big(T,Q,\alpha, \beta, u \mapsto (Q \times Q) \setminus R(u)\big)$ 
is an NLC-decomposition of the complement of $G$.

\section{Overview of the proofs} \label{sec:overview}

\subsection{Bounded treewidth}
We give a high-level idea of our sample compression scheme for graphs of bounded treewidth.
The key property that we exploit is that for a graph $G$ with a tree decomposition $\big(T,(W_x\mid x\in V(T))\big)$, for each $y\in V(T)$, the bag $W_y$ separates in $G$ the sets
$\bigcup_{x\in V(C)} W_x$ for all components of $C$ of $T-y$.
We begin by discussing how small separations aid the design of efficient sample compression schemes.

Let $G$ be a graph, let $H$ be the hypergraph of balls in $G$ and let $(A_1, A_2)$ be a separation of $G$.
Let $S = A_1 \cap A_2$ be the separator of this separation, and let $k = |S|$. 
%\piotr{I am here.}
Every $A_1$-$A_2$ path in $G$ intersects $S$, so for every vertex $v \in A_1$, the distances between $v$ and the vertices in $A_2$ are entirely characterized by the distances between $v$ and the $k$ vertices in $S$.
Namely, for every $v \in A_1$ and $u \in A_2$, we have $\dist_G(v,u) = \min_{w \in S} \left(\dist_G(v,w) + \dist_G(w,u)\right)$.
In particular, when a vertex $v \in A_1$ is not known, but we know the values $\dist_G(v,w)$ for every $w \in S$, we can retrieve $\dist_G(v,u)$ for every $u \in A_2$.
In other words, we can compress the information about all distances between $v$ and the vertices in $A_2$ into the information about all distances between $v$ and the $k$ vertices in $S$.
Suppose that $(X^+,X^-)$ is a sample of $H$ realized by $B_G(c,r)$ for some $c \in V(G)$ and some integer $r$, such that $c \in A_1$ and $X^+ \cup X^- \subseteq A_2$.
A natural approach is to encode the distances between $c$ and the vertices of $S$ in $G$.
However, we do not see a direct way of encoding these distances using only the vertices of $X^+ \cup X^-$, and thus we take a somewhat dual approach.
Namely, for each vertex $w \in S$, we will encode a lower bound and an upper bound on the distance between $w$ and the boundary of a ball realizing the sample.
%Namely, for each vertex $w \in S$, we will encode the minimal distance from $w$ to vertices in $A_2$ that have to be covered by a ball realizing the sample, and symmetrically the maximal such distance.
More precisely, for every $w \in S$, we define
\begin{align*}
        r^+(w) &= \max\set{\dist_G(w,x)\mid x\in (X^+\cup X^-) \cap B_G(v,r-\dist_G(c,w))},\\
        r^-(w) &= \min\set{\dist_G(w,x)\mid x\in X^-}.
\end{align*}
The crucial observation (\cref{lemma:good-balls-realize}) is that for a ball $B_G(c',r')$ for some $c' \in A_1$ and some integer $r'$ to realize $(X^+,X^-)$, it suffices that
 \[
        A_2 \cap \bigcup_{w \in S} B_G(w, r^+(w)) \subseteq 
        A_2 \cap B_G(c', r') \subseteq 
        A_2 \cap \bigcup_{w \in S} B_G(w, r^-(w) - 1).
\]
Therefore, assuming that the compressor encoded the values $r^+(w)$ and $r^-(w)$ for every $w \in S$, it is enough to find a ball satisfying the sufficient condition, which can be done by trying all possibilities for $c' \in A_1$ and an integer $r'$.
Note that the ball $B_G(c,r)$ satisfies the condition (\cref{lemma:initial-ball-realizes}), hence, indeed having the values $r^+(w)$ and $r^-(w)$ for every $w \in S$, we can find a ball realizing the sample.
Finally, note that to encode the values $r^+(w)$ and $r^-(w)$ for $w \in S$ it suffices to remember the witnesses in the respective definitions.

Now assume that the treewidth of $G$ is at most $t$ and fix a tree decomposition $\big(T, (W_y \mid y \in V(T))\big)$ of $G$ of width at most $t$ where $T$ is a rooted tree.
For every vertex $v \in V(G)$, we fix a ``home node'' $h(v) \in V(T)$ such that $v \in W_{h(v)}$.

Using the tree decomposition of $G$, we would like to find a separation $(A_1,A_2)$ of $G$ such that $c \in A_1$ and $X^+ \cup X^- \subset A_2$.
The details of this process are described in~\Cref{sec:subtrees}.
In short, we consider the set $Z = \{h(x) \mid x \in X^+ \cup X^-\}$ of all home nodes of vertices of $X^+ \cup X^-$.
Observe that every node $z \in \LCA(T, Z)$ can be described in terms of two elements $x_1, x_2 \in X^+ \cup X^-$ such that $z = \lca(T, h(x_1), h(x_2))$.
If $h(c) \in \LCA(T,Z)$, then we use this observation to encode $h(c)$ and the required separation is $(W_{h(c)},\bigcup_{y \in V(T) \setminus\{h(c)\}}W_y)$.
Thus, we may assume that $h(c) \notin \LCA(T, Z)$.
The component $C$ of $T - \LCA(T, Z)$ which contains $h(c)$ has at most two neighbors in $\LCA(T, Z)$.
It follows that $C$ is determined by two elements of $\LCA(T,Z)$, and these elements can be encoded by at most four elements of $X^+ \cup X^-$ as observed above.
Finally, we set 
\[(A_1, A_2) = \left(\bigcup_{y \in V(C)}W_y, \bigcup_{y \notin V(C)}W_y\right).\]
Note that indeed $c \in A_1$ and $X^+ \cup X^- \subset A_2$.

Summarizing, the compression phase consists of two parts.
First, we find a suitable separation $(A_1,A_2)$ of order roughly $2t$ and we encode it using few vertices of $X^+ \cup X^-$.
Second, for every $w \in A_1 \cap A_2$, we encode $r^+(w)$ and $r^-(w)$ using two elements of $X^+ \cup X^-$.
In total, we use $\Oh(t)$ elements.
Note that $r^+(w)$ and $r^-(w)$ have to be \q{assigned} to the element $w$.
This is when we use the fact that we work with \emph{array} sample compression schemes and we lose a logarithmic factor in the final result (\Cref{thm:tw-first}).

\subsection{Bounded cliquewidth}

Our approach for graphs of bounded NLC-width is very similar to our approach for graphs of bounded treewidth. 
The graphs of bounded NLC-width are the graphs that can be recursively decomposed along simple edge cuts, in the sense that there is a small diversity of neighborhoods between the two sides of the cut. To understand these graphs from a metric point of view, we then have to understand how such edge cuts interact with distances.

Let $G$ be a graph and let $(A_1, A_2)$ be a bipartition of $G$ such that the vertices in $A_1$ can be partitioned into sets $V_1, \ldots, V_k$ such that any two vertices in the same $V_i$ have the same neighborhood in $A_2$.
Then, every $A_1$-$A_2$ path in $G$ leaves $A_1$ via some $V_i$. 
Since all vertices in $V_i$ have the same neighborhood in $A_2$, we get that for every vertex $v \in A_1$, the distances between $v$ and the vertices in $A_2$ are entirely characterized by the $k$ distances between $v$ and the sets $V_1, \dots, V_k$, see~\Cref{lemma:separator_with_few_types}.
Again, this ``compresses'' the information of all distances between $v$ and the vertices in $A_2$ to the information of all $k$ distances between $v$ and the sets $V_1, \dots, V_k$.
However, as before, we can not use the vertices of the sample to encode this information.
Instead, we exploit a similar workaround as in the case of treewidth.
Up to some extra technicalities, the final approach for constructing sample compression schemes for hypergraphs of balls of graphs of bounded cliquewidth is similar to the bounded treewidth case.

\section{Encoding subtrees}\label{sec:subtrees}

In this section, we prove \Cref{lemma:encode-subtree}, which will be a crucial tool
in the proofs of \Cref{thm:tw,thm:cw}.
The following lemma is folklore, see e.g. \cite[Lemma 8]{DHHJLMMRW25}. 

\begin{lemma}\label{lemma:increase_Z_in_a_tree}
    Let $T$ be a rooted tree,
    and let $Z$ be a set of the nodes of $T$.
    Every component $C$ of $T-\LCA(T,Z)$ satisfies $|N_T(V(C))| \leq 2$.
\end{lemma}

Consider the following setting.
Let $T$ be a rooted binary tree and let $S(T)$ be the set of subtrees of $T$.
For each $Z \subseteq V(T)$, let $S(T, Z)$ be the set of all components of $T - \LCA(T,Z)$ and all one-node subtrees given by the nodes in $\LCA(T,Z)$.
Our goal is to encode the subtrees of $T$, using nodes of $T$, in such a way that for each $Z \subseteq V(T)$, each subtree $T' \in S(T, Z)$ can be encoded with few elements of $Z$.
A simple way to do it is the following.
We define a function $\varphi\colon \binom{V(T)}{\leq 3} \times [16] \rightarrow S(T)$ such that for every $Z \subset V(T)$, the image $\varphi\left(\binom{Z}{\leq 3} \times [16] \right)$ contains $S(T,Z)$.
Let $Z'$ be a set of at most three nodes in $T$.
Note that $|\LCA(T,Z')| \le 5$.
In particular, $|S(T,Z')| \leq 16$ as $T$ is a binary tree.
We may now enumerate the components of $S(T,Z')$ and assign them as values of $\varphi(Z',i)$ for each $i \in [|S(T,Z')|]$.
This would work well in the usual setup of sample compression schemes.
However, since we chose to use \emph{array} sample compression schemes, we need to adjust the function $\varphi$ accordingly.
Namely, we need a function $(V(T) \cup \{\emptyentry\})^M \rightarrow S(T)$ for some small integer $M$.
One can obtain a constant $M$ from the reasoning above by encoding the numbers in $[16]$ using the empty entry symbol $\emptyentry$.
In the lemma below, we are more careful, and we get $M = 3$.

\begin{lemma}\label{lemma:encode-subtree}
    Let $T$ be a rooted binary tree. Let $S(T)$ be the set of subtrees of $T$. There exists a map $\varphi\colon (V(T) \cup \{\emptyentry\})^3 \to S(T)$ such that 
    for every set $Z \subseteq V(T)$, 
    the image $\varphi((Z \cup \{\emptyentry\})^3)$ contains each 
    component of $T-\LCA(T, Z)$ and each one-node subtree with the node in $\LCA(T, Z)$.
\end{lemma}

\begin{proof}
    Let $r$ be the root of $T$. 
    We define $\varphi$ as follows.
    Let $z_1,z_2,z_3 \in V(T)$.
    \begin{itemize}
        \item If there exist distinct $y_1, y_2 \in V(T)$ such that $\LCA(T, \{z_1, z_2, z_3\}) = \set{z_1,z_2,z_3,y_1, y_2}$ and there is a (unique) component $C$ of $T-\{y_1, y_2\}$ with $N_T(V(C)) = \{y_1, y_2\}$, we set $\varphi(z_1, z_2, z_3) = C$.
        \item We set $\varphi(z_1, z_2, \emptyentry) = T[\{\lca(T,z_1, z_2)\}]$. 
        \item If $r\not\in\LCA(T,\{z_1, z_2\})$, then there exists a (unique) component $C$ of $T - \LCA(z_1, z_2)$ that contains $r$. We set $\varphi(\emptyentry, z_1, z_2) = C$. 
        \item If $z_1$ has a left child in $T$ then let $C$ be the component of $T - z_1$ that contains the left child of $z_1$. We set $\varphi(z_1,\emptyentry, \emptyentry) = C$. 
        \item If $z_1$ has a right child in $T$ then let $C$ be the component of $T - z_1$ that contains the right child of $z_1$. We set $\varphi(\emptyentry, \emptyentry,z_1) = C$.
        \item In all other cases, we set $\varphi(\cdot, \cdot, \cdot) = T$. In particular, $\varphi(\emptyentry, \emptyentry, \emptyentry) = T$.
    \end{itemize}
We argue that $\varphi$ satisfies the desired property. Let $Z \subseteq V(T)$ and let $C$ be either a component of $T-\LCA(T, Z)$ or a one-node graph with the node in $\LCA(T, Z)$.
We show that there exist $z_1',z_2',z_3' \in Z \cup \{\emptyentry\}$ such that $\varphi(z_1',z_2',z_3') = C$.
Suppose first that $V(C)$ consists of one node $z$ in $\LCA(T, Z)$. 
Then, there exist $z_1, z_2 \in Z$ such that $z = \lca(T,z_1, z_2)$ (possibly, $z_1 = z_2$), and so, $ \varphi(z_1, z_2, \emptyentry) = C$.
Next, suppose that $C$ is a component of $T-\LCA(T, Z)$.
By~\cref{lemma:increase_Z_in_a_tree}, we have $|N_T(V(C))| \leq 2$.
If $|N_T(V(C))| = 0$, then $Z = \emptyset$ and $C = T$ so $ \varphi(\emptyentry, \emptyentry, \emptyentry) = C$.

Suppose that $|N_T(V(C))| = 1$ and let $N_T(V(C)) = \{z\}$. Then, $C$ is a component of $T-\{z\}$. If $C$ contains $r$ then $z \neq r$ and there exist $z_1, z_2 \in Z$ such that $z = \lca(T, z_1, z_2)$, and then $C$ is the unique component of $T - \lca(T, z_1, z_2)$ that contains $r$, so $\varphi(\emptyentry, z_1, z_2) = C$.
If $C$ does not contain $r$, then $C$ contains either the left child of $z$ or the right child of $z$.
Assume that $C$ contains the left child $z_1$ of $z$. 
No descendant of $z_1$ is a node of $Z$ since $V(C)$ is disjoint from $Z$, and therefore $z \in \LCA(T, Z)$ implies $z \in Z$.
Therefore, $C$ is the unique component of $T - z$ that contains the left child of $z$, so $\varphi(z,\emptyentry, \emptyentry) = C$.
If $C$ contains the right child of $z$, then by a symmetric argument, we have $\varphi(\emptyentry, \emptyentry,z) = C$.

Finally, suppose that $|N_T(V(C))| = 2$ and let $N_T(V(C)) = \{y_1, y_2\}$. 
Since $C$ is a component of $T-\LCA(T, Z)$, we have $y_1, y_2 \in \LCA(T, Z)$. 
We claim that $y_1$ and $y_2$ are comparable in $T$. 
Suppose otherwise. 
Then, for $z = \lca(T, y_1, y_2)$, we have $z \notin \{y_1, y_2\}$ and $z \in \LCA(T, Z)$. 
Since $y_1, y_2 \in N_T(V(C))$, all the internal nodes of the $y_1$-$y_2$ path in $T$ are in $V(C)$.
However, this path contains $z$, which contradicts that $V(C)$ is disjoint from $\LCA(T, Z)$. 
This concludes the proof that $y_1$ and $y_2$ are comparable in $T$.
Therefore, up to renaming $y_1$ and $y_2$, there exist $z_1, z_2, z_3 \in Z$ such that $\lca(T,z_1, z_2) = y_1$ and $\lca(T, z_1, z_3) = y_2 = \lca(T, z_2, z_3)$. 
Thus, $\LCA(T, \{z_1,z_2,z_3\}) = \{z_1,z_2,z_3,y_1, y_2\}$ and $C$ is the component of $T-\{y_1, y_2\}$ such that $N_T(V(C)) = \{y_1, y_2\}$. 
It follows that $\varphi(z_1, z_2, z_3) = C$ as desired.
This completes the proof.
\end{proof}

\section{Bounded treewidth}\label{sec:treewidth}

%\subsection{Preliminaries}

An \defin{instance} is a tuple $(G,(X^+,X^-),c,r)$ where $G$ is a graph,
$(X^+,X^-)$ is a sample of the hypergraph of balls in $G$, $c$ is a vertex of $G$, and $r$ is an integer such that $B_G(c,r)$ realizes $(X^+,X^-)$.

Recall that $\min \emptyset = \infty$ and $\max \emptyset = -1$. 
For an instance $\mathcal{I}=(G,(X^+,X^-),c,r)$, 
for each vertex $v$ in $G$, 
let
\begin{align*}
        \defin{\text{$r^+_{\mathcal{I}}(v)$}} &= \max\set{\dist_G(v,x)\mid x\in X \cap B_G(v,r-\dist_G(c,v))},\\
        \defin{\text{$r^-_{\mathcal{I}}(v)$}} &= \min\set{\dist_G(v,x)\mid x\in X^-}.
\end{align*}
% \begin{align*}
%         X_{\mathcal{I}}(v) &= X \cap B_G(v,r-\dist_G(c,v)),\\
%         r^+_{\mathcal{I}}(v) &= 
%                 \begin{cases}
%                     \max\set{\dist_G(v,x)\mid x\in X_\calI(v)}&\textrm{if $X_{\mathcal{I}}(v) \neq \emptyset$} \\
%                     -1 &\textrm{otherwise,}
%                 \end{cases}\\
%         r^-_{\mathcal{I}}(v) &= 
%                 \begin{cases}
%                     \min\set{\dist_G(v,x)\mid x\in X^-}\ \ &\textrm{if $X^- \neq \emptyset$}\\
%                     \infty & \textrm{otherwise.}
%                 \end{cases}
% \end{align*}
Note that $r^+_{\mathcal{I}}(v) \leq r - \dist_G(c,v) \leq r_{\mathcal{I}}^-(v)-1$.

As we explained in \cref{sec:overview}, given an instance $\mathcal{I}=(G,(X^+,X^-),c,r)$ and a separation $(A_1,A_2)$ of $G$ such that $c \in A_1$ and $X^+\cup X^- \subseteq A_2$,  
%\jedrzej{This should be $A_2$?} \piotr{Yes.} 
one can verify whether a ball centered in a vertex $c' \in A_1$ realizes $(X^+,X^-)$ by only inspecting distances between $A_1\cap A_2$ and $\set{c, c'}\cup X^+\cup X^-$ in $G$.
This is formally captured in the next two lemmas:
\Cref{lemma:good-balls-realize} gives a sufficient condition on balls $B(c',s)$ with $c' \in A_1$
to ensure that $B(c',s)$ realizes $(X^+,X^-)$,
while \Cref{lemma:initial-ball-realizes} states that there is a ball (namely $B_G(c,r)$) satisfying this condition.

%The next lemma gives a sufficient condition for when a ball in a graph $G$ realizes a sample from the hypergraph of balls in $G$. 
%The next lemma shows that this condition is always satisfied by some ball in the graph. 

\begin{lemma}\label{lemma:good-balls-realize}
    Let $\mathcal{I}=(G,(X^+,X^-),c,r)$ be an instance,
    let $X = X^+ \cup X^-$,
    let $(A_1,A_2)$ be a separation of $G$ such that $c\in A_1$ and $X \subseteq A_2$, and let $S = A_1 \cap A_2$. 
    Let $c' \in A_1$ and let $s$ be an integer such that 
    \[
        A_2 \cap \bigcup_{v \in S} B_G(v, r^+_{\mathcal{I}}(v)) \subseteq 
        A_2 \cap B_G(c', s) \subseteq 
        A_2 \cap \bigcup_{v \in S} B_G(v, r^-_{\mathcal{I}}(v) - 1) .
    \]
    Then, $B_G(c', s)$ realizes the sample $(X^+,X^-)$.
\end{lemma}

\begin{proof}
It suffices to prove that $X^+ \subseteq B_G(c', s)$ and $X^- \cap B_G(c', s) = \emptyset$.
We start with the former.
Let $x \in X^+$. Since $B_G(c, r)$ realizes the sample $(X^+, X^-)$, we have $x \in B_G(c,r)$, and so, $\dist_G(c, x) \leq r$.
Let $P$ be a shortest $c$-$x$ path in $G$.
Since $c \in A_1$ and $x \in X \subseteq A_2$, the path $P$ intersects $S$ in some vertex, say $v$.
Then, $\dist_G(c, v) + \dist_G(v, x) = \dist_G(c, x) \leq r$ so $\dist_G(v, x) \leq r-\dist_G(c, v)$.
Thus, $x \in X \cap B_G(v, r-\dist_G(c, v))$ so $r^+_{\mathcal{I}}(v) \geq \dist_G(v, x)$.
Therefore, $x \in A_2 \cap B_G(v, r^+_{\mathcal{I}}(v)) \subseteq B_G(c', s)$.

Next, we prove that $X^- \cap B_G(c', s) = \emptyset$.
By contradiction, suppose that there exists $x \in X^- \cap B_G(c', s)$.
Then, $x \in A_2 \cap B_G(c', s) \subseteq A_2 \cap \bigcup_{v \in S}B_G(v, r^-_{\mathcal{I}}(v) - 1)$,
so there exists $v \in S$ such that $x \in B_G(v, r^-_{\mathcal{I}}(v) - 1)$.
Since $x \in X^-$, this contradicts the definition of $r^-_{\mathcal{I}}(v)$.
\end{proof}

\begin{lemma}\label{lemma:initial-ball-realizes}
    Let $\mathcal{I}=(G,(X^+,X^-),c,r)$ be an instance,
    let $X = X^+ \cup X^-$,
    let $(A_1,A_2)$ be a separation of $G$ such that $c\in A_1$ and $X \subseteq A_2$, and let $S = A_1 \cap A_2$. 
    Then, 
    \[
        A_2 \cap \bigcup_{v \in S}B_G(v, r^+_{\mathcal{I}}(v)) \subseteq A_2  \cap  B_G(c, r) \subseteq 
        A_2 \cap \bigcup_{v \in S}B_G(v, r^-_{\mathcal{I}}(v) - 1).
    \]
\end{lemma}

\begin{proof}
    We first prove that $B_G(v, r^+_{\mathcal{I}}(v)) \subseteq B_G(c, r)$ for every $v \in S$, which implies that $A_2 \cap \bigcup_{v \in S}B_G(v, r^+_{\mathcal{I}}(v)) \subseteq A_2 \cap B_G(c, r)$.
    Let $v \in S$.
    If $r^+_{\mathcal{I}}(v) = -1$ then $B_G(v, r^+_{\mathcal{I}}(v)) = \emptyset \subseteq B_G(c, r)$. 
    Otherwise, by the definition of $r^+_{\mathcal{I}}(v)$, we have $r^+_{\mathcal{I}}(v) \leq r-\dist_G(c,v)$ so $B_G(v, r^+_{\mathcal{I}}(v)) \subseteq B_G(v, r-\dist_G(c,v)) \subseteq B_G(c, r)$.

    We then prove that $A_2 \cap B_G(c, r) \subseteq A_2 \cap \bigcup_{v \in S} B_G(v, r^-_{\mathcal{I}}(v) - 1) $.
    Let $u \in A_2 \cap B_G(c, r)$, and let $P$ be a shortest path from $c$ to $u$. 
    Since $(A_1, A_2)$ is a separation of $G$ and $c \in A_1$, we have that $P$ contains a vertex $v \in S$. 
    By contradiction, suppose that $u \notin B_G(v, r^-_{\mathcal{I}}(v) - 1)$.
    Then, $r^-_{\mathcal{I}}(v) \neq \infty$ so there exists a vertex $u' \in X^-$ such that $\dist_G(v, u') = r^-_{\mathcal{I}}(v) \leq \dist_G(v, u)$. 
    Then, $\dist_G(c, u') \leq \dist_G(c, v) + \dist_G(v, u') \leq \dist_G(c, v) + \dist_G(v, u) = \dist_G(c, u) \leq r$ (where the equality holds since $P$ is a shortest path from $c$ to $u$, and visits $v$).
    Therefore, we have $u' \in X^- \cap B_G(c, r)$, which contradicts that $B_G(c, r)$ realizes $(X^+, X^-)$.
\end{proof}

We proceed with a proof of~\cref{thm:tw} which we restate below.

\thmtw*

\begin{proof}
    Let $t$ be a positive integer and let $G$ be a graph of treewidth at most $t$.
    Let $H$ be the hypergraph of balls in $G$.
    We fix a tree decomposition $\big(T, (W_y \mid y \in V(T))\big)$ of $G$ of width $t$ where $T$ is a rooted binary tree.
    Thus, each bag and each adhesion is of size at most $t+1$. 
    Let $r$ be the root of $T$, and for every vertex $v \in V(G)$, we fix a node $h(v) \in V(T)$ such that $v \in W_{h(v)}$.
    Additionally, set $h(\emptyentry) = \emptyentry$.
    We also fix a map $\varphi\colon (V(T) \cup \{\emptyentry\})^3 \to S(T)$ which satisfies the condition of \cref{lemma:encode-subtree},
    and an arbitrary total order $\prec$ on $V(G)$. 

    \header{Compression.}
            Let $(X^+,X^-)$ be a sample of $H$.
            Let $X = X^+ \cup X^-$ and 
            fix $c\in V(G)$ and an integer $r$ such that $B_G(c, r) \cap X = X^+$. 
            Let $\mathcal{I} = (G,(X^+,X^-),c,r)$.
            Consider the set $Z = \{h(x) \mid x \in X\}$. 
            If $h(c) \in \LCA(T, Z)$ then set $C = T[\{h(c)\}]$, and otherwise let $C$ be the component of $T - \LCA(T, Z)$ which contains $h(c)$.
            In both cases, there exist $z_1, z_2, z_3 \in Z \cup \{\emptyentry\}$ such that $\varphi(z_1, z_2, z_3) = C$.
            For each $i \in [3]$, pick an arbitrary $w_i \in X \cup \{\emptyentry\}$ such that $z_i = h(w_i)$.
            Thus, if $h(c) \in \LCA(T, Z)$ then $\varphi(h(w_1), h(w_2), h(w_3)) = T[\{h(c)\}]$, and otherwise $\varphi(h(w_1), h(w_2), h(w_3))$ is the component of $T - \LCA(T, Z)$ which contains $h(c)$. Let
            \[
            (A_1,A_2) = \begin{cases}
            (W_y,V(G))&\textrm{if $V(C)=\set{y}$ for some $y\in V(T)$,}\\
            \left(\bigcup_{y \in V(C)}W_y, \bigcup_{y \notin V(C)}W_y\right)&\textrm{otherwise.}
            \end{cases}
            \]
            Let $S = A_1 \cap A_2$.
            If $C$ consists of a single node $y$, then $S = W_y$ so $|S| \leq t+1$.
            Otherwise, $C$ is a component of $T - \LCA(T, Z)$ so by \cref{lemma:increase_Z_in_a_tree}, we have $|N_T(V(C))| \leq 2$, and thus $S$ is the union of at most two adhesion sets, so $|S| \leq 2(t+1)$. 
            Let $\ell = |S| \leq 2t+2$ and let $v_1, \ldots, v_{\ell}$ be the vertices of $S$, ordered according to $\prec$.
            For each $i \in [\ell]$, if $r^+_{\mathcal{I}}(v_i) \neq -1$, let $x_i^+$ be a vertex in $X$ such that $\dist_G(v_i, x_i^+) = r^+_{\mathcal{I}}(v_i)$, and let $x_i^+ = \emptyentry$ otherwise.
            Similarly, if $r^-_{\mathcal{I}}(v_i) \neq \infty$, let $x_i^-$ be a vertex in $X$ such that $\dist_G(v_i, x_i^-) = r^-_{\mathcal{I}}(v_i)$, and let $x_i^- = \emptyentry$ otherwise.
            For each $i \in [2t+2] \setminus [\ell]$, let $x_i^+ = x_i^- = \emptyentry$.
            Finally, set
            \[\kappa(X^+, X^-) = (w_1, w_2, w_3, x_1^+, x_1^-, \ldots, x_{2t+2}^+, x_{2t+2}^-).\]

    \header{Reconstruction.}
            Consider an input vector for the reconstructor $\rho$:
            \[(w_1, w_2, w_3, x_1^+, x_1^-, \ldots, x_{2t+2}^+, x_{2t+2}^-) \in (V(G)\cup\set{\emptyentry})^{4t+7}.
            \]
            Let $C$ be the subtree of $T$ output by 
            $\varphi(h(w_1), h(w_2), h(w_3))$. Let
            \[
            (A_1,A_2) = \begin{cases}
            (W_y,V(G))&\textrm{if $V(C)=\set{y}$ for some $y\in V(T)$,}\\
            \left(\bigcup_{y \in V(C)}W_y, \bigcup_{y \notin V(C)}W_y\right)&\textrm{otherwise.}
            \end{cases}
            \]
            Let $S = A_1 \cap A_2$.
            If $S$ contains more than $2t+2$ elements, then it does not mattter what is the output of $\rho$, so $\rho$ returns the empty ball. 
            Suppose now that $|S| = \ell \leq 2t+2$, and let $v_1, \ldots, v_{\ell}$ be the vertices of $S$, ordered according to $\prec$.
            For each $i \in [\ell]$, if $x_i^+ \neq \emptyentry$, let $r_i^+ = \dist_G(v_i, x_i^+)$, and otherwise, let $r_i^+ = -1$. Similarly, if $x_i^- \neq \emptyentry$, let $r_i^- = \dist_G(v_i, x_i^-)$, and otherwise, let $r_i^- = \infty$.
            The reconstructor $\rho$ verifies whether there exists a ball $B_G(c', s)$ with $c' \in A_1$ such that
            \[
                A_2 \cap \bigcup_{i \in [\ell]} B_G(v_i, r_i^+) 
                \subseteq A_2 \cap B_G(c', s) \subseteq
                A_2 \cap \bigcup_{i \in [\ell]} B_G(v_i, r_i^- - 1).
            \]
            If yes, then $\rho$ returns such a ball, and otherwise it does not matter what is the output of $\rho$, so $\rho$ returns the empty ball. 

    \header{Proof of correctness.}
    We now prove that $(\kappa,\rho)$ 
    is a proper array sample compression scheme of $H$. 
    By definition, $\rho$ only returns balls in $G$.
    Clearly, for each sample $(X^+, X^-)$ of $H$, 
    we get $\kappa(X^+,X^-)$ which is an element of $(X\cup\set{\emptyentry})^{4t+7}$.
    We need to argue that $\rho(\kappa(X^+, X^-))$ realizes $(X^+, X^-)$.

    Let $(X^+,X^-)$ be a sample of $H$ and let
    \[\kappa(X^+, X^-) = (w_1, w_2, w_3, x_1^+, x_1^-, \ldots, x_{2t+2}^+, x_{2t+2}^-).\]
    Let $Z = \{h(x) \mid x \in X\} \subseteq V(T)$ and $C = \varphi(h(w_1), h(w_2), h(w_3))$. 
    Let $c\in V(G)$ and an integer $r$ be as they were fixed in the definition of $\kappa(X^+, X^-)$. 
    Thus, $B_G(c, r)$ realizes $(X^+, X^-)$.
    Observe that if $h(c) \in \LCA(T, Z)$, then $C = T[\{h(c)\}]$, and otherwise $C$ is the component of $T - \LCA(T, Z)$ which contains $h(c)$.

    By construction, the compressor and the reconstructor consider the same separation $(A_1, A_2)$ and the same set $S=A_1\cap A_2$.
    Then, $|S| \leq 2t+2$ (as we already argued in the definition of the compressor), and both the compressor and the reconstructor consider the vertices $v_1, \ldots, v_{\ell}$ of $S$, for $\ell=|S|$, ordered according to $\prec$.
    
    Let $\mathcal{I}=(G,(X^+,X^-),c,r)$.
    It follows immediately from construction that for every $i \in [\ell]$, the values $r_i^+$ and $r_i^-$ computed by the reconstructor satisfy $r_i^+ = r^+_{\mathcal{I}}(v_i)$ and $r_i^- = r^-_{\mathcal{I}}(v_i)$.

    We now claim that the separation $(A_1, A_2)$ satisfies $c \in A_1$ and $X \subseteq A_2$. 
    Indeed, if $C$ consists of a single node $y$ then $(A_1, A_2) = (W_y, V(G))$, so $y = h(c)$ and $c \in W_y = A_1$, and trivially $X \subseteq V(G) = A_2$.
    Otherwise, $(A_1, A_2) = \left(\bigcup_{y \in V(C)}W_y, \bigcup_{y \notin V(C)}W_y\right)$.
    In that case, $C$ is the component of $T - \LCA(T, Z)$ which contains $h(c)$.
    Then, $c \in W_{h(c)} \subseteq \bigcup_{y \in V(C)}W_y = A_1$.
    Furthermore, $C$ is disjoint from $\LCA(T, Z)$, therefore $Z \subseteq V(T) - V(C)$.
    Thus, for every $x \in X$, we have $x \in W_{h(x)} \subseteq \bigcup_{y \in Z}W_y \subseteq \bigcup_{y \notin V(C)}W_y = A_2$.
    This proves that $c \in A_1$ and $X \subseteq A_2$.

    By \cref{lemma:initial-ball-realizes}, it follows that
    \[
       A_2 \cap \bigcup_{i \in [\ell]} B_G(v_i, r_i^+)  
        \subseteq A_2 \cap B_G(c, r) \subseteq 
        A_2 \cap \bigcup_{i \in [\ell]} B_G(v_i, r_i^- - 1).
    \]
    Thus, the reconstructor returns a ball $B_G(c', s)$ such that $c' \in A_1$ and 
    \begin{equation}\label{eq:to-satisfy}\tag{$\star$}
        A_2 \cap \bigcup_{i \in [\ell]} B_G(v_i, r_i^+) 
        \subseteq A_2 \cap B_G(c', s) \subseteq 
        A_2 \cap \bigcup_{i \in [\ell]} B_G(v_i, r_i^- - 1).
    \end{equation}
    By \cref{lemma:good-balls-realize}, we deduce that $B_G(c',s) = \rho(\kappa(X^+,X^-))$ realizes $(X^+, X^-)$.
\end{proof}

Let us remark that the proof of~\Cref{thm:tw} gives something slightly stronger.
In the constructed sample compression scheme, the compressor $\kappa$ fixes a ball $B_G(c,r)$, which is consistent with a given sample $(X^+,X^-)$.
The reconstructor returns any ball $B_G(c', s)$ that satisfies \cref{eq:to-satisfy}, and we argued that $B_G(c, r)$ satisfies \cref{eq:to-satisfy}.
Thus, if the reconstructor returns a ball $B_G(c', s)$ which satisfies \cref{eq:to-satisfy} and with $s$ minimal, then $s \leq r$.
We deduce the following statement.

\begin{corollary}\label{cor:tw}
    For all positive integers $t$ and $r$ and every graph $G$ of treewidth at most $t$, 
    the hypergraph of balls of radius at most $r$ in $G$ has a proper array sample compression scheme of size $4t+7$.
\end{corollary}

\section{Bounded cliquewidth}\label{sec:cliquewidth}
In this section, we prove~\Cref{thm:cw}, i.e.\ for every positive integer $t$ and for every graph $G$ of cliquewidth at most $t$, the hypergraph of balls in $G$ has a proper array sample compression scheme of size $4t+3$.
The proof follows a similar approach as the proof of the analogous result for treewidth (\Cref{thm:tw}), although the details are a bit more technical.

Let $G$ be a graph and let $\set{C, D}$ be a partition\footnote{A \defin{partition} of a set $X$ is a collection of pairwise disjoint subsets of $X$, whose union is $X$. We allow parts to be empty. An \defin{ordered partition} of a set $X$ is a finite sequence of pairwise disjoint subsets of $X$, whose union is $X$.} of $V(G)$.
A \defin{$D$-twin partition} of $C$ is a partition $(V_1, \ldots, V_k)$ of $C$ such that for every $i \in [k]$, for all $u,v \in V_i$, we have $N_G(u) \cap D = N_G(v) \cap D$. 

\begin{lemma}\label{lemma:separator_with_few_types}
    Let $G$ be a graph,
    let $(C,D)$ be an ordered partition of $V(G)$,
    let $c_0 \in C$ and $d_0 \in D$,
    and let $(V_1, \dots, V_k)$ be a $D$-twin partition of $C$.
    Then,
    \[
        \dist_G(c_0,d_0) = \min_{i \in [k]} \left(\dist_{G}(c_0,V_i) + \dist_{G[V_i \cup D]}(V_i, d_0)\right).
    \]
\end{lemma}

\begin{proof}
    Consider a shortest $c_0$-$d_0$ path $P$ in $G$.
    Then, $P$ intersects $C = \bigcup_{i \in [k]} V_i$.
    Let $w$ be the vertex of $P$ in $\bigcup_{i \in [k]} V_i$ that is the closest to $d_0$ in $P$.
    Let $i \in [k]$ be such that $w \in V_i$.
    Let $P_1$ be the $c_0$-$w$ subpath of $P$ and let $P_2$ be the $w$-$d_0$ subpath of $P$.
    We have,
    \[
        \dist_{G}(c_0,V_i) + \dist_{G[V_i \cup D]}(V_i, d_0) 
        \leq \length(P_1) + \length(P_2) 
        = \length(P) 
        = \dist_G(c_0,d_0).
    \]
    In particular, $\dist_G(c_0,d_0) \geq \min_{i \in [k]} \left(\dist_{G}(c_0,V_i) + \dist_{G[V_i \cup D]}(V_i, d_0)\right)$.

    For the other inequality,
    consider some $i \in [k]$.
    Let $P_1$ be a shortest $c_0$-$V_i$ path in $G$ and let
    $P_2$ be a shortest $V_i$-$d_0$ path in $G[V_i \cup D]$.
    Let $u$ be the vertex of $P_1$ in $V_i$, let $v$ be the vertex of $P_2$ in $V_i$, and let $w$ be the neighbor of $v$ in $P_2$.
    Since $u, v \in V_i$, we have $N_G(u) \cap D = N_G(v) \cap D$.
    In particular, $w \in N_G(u)$, and we may define the path $P_2'$ to be obtained from $P_2$ by replacing the endpoint $v$ by $u$.
    Note that $P_1 \cup P_2'$ is a $c_0$-$d_0$ path in $G$.
    Finally, we obtain
    \begin{align*}
        \dist_G(c_0,d_0) \leq \length(P_1 \cup P_2') &= \length(P_1) + \length(P_2')\\ 
        &= \length(P_1) + \length(P_2) = \dist_{G}(c_0,V_i) + \dist_{G[V_i \cup D]}(V_i, d_0).
    \end{align*}
    Since the above inequality holds for every $i \in [k]$, the proof is completed.
\end{proof}

%\piotr{I am here.}

Let $\mathcal{I} = (G,(X^+,X^-),c,r)$ be an instance and let $X = X^+ \cup X^-$.
Let $(C, D)$ be an ordered partition of $V(G)$ such that $c \in C$. 
Let $(U_1, \ldots, U_k)$ be a $C$-twin partition of $D$ and let $(V_1, \ldots, V_{\ell})$ be a $D$-twin partition of $C$.
Recall that $\min \emptyset = \infty$ and $\max \emptyset = -1$.
For all $i \in [k]$ and $j \in [\ell]$, we define
\begin{align*}
        \defin{\text{$r^+_{(\mathcal{I}, (C, D))}(U_i)$}} &= \max\set{\dist_G(U_i, x) \mid x\in X^+\cap B_G(U_i,r-\dist_{G[C \cup U_i]}(c,U_i))},\\
        \defin{\text{$r^-_{(\mathcal{I}, (C, D))}(U_i)$}} &= \min\set{\dist_G(U_i,x)\mid x\in X^-},\\
        \defin{\text{$r^+_{(\mathcal{I}, (C, D))}(V_j)$}} &= \max\set{\dist_{G[D \cup V_j]}(V_j,x)\mid x\in X^+\cap B_{G[D \cup V_j]}(V_j,r-\dist_G(c,V_j))}, \\
        \defin{\text{$r^-_{(\mathcal{I}, (C, D))}(V_j)$}} &= \min\set{\dist_{G[D \cup V_j]}(V_j,x)\mid x\in X^- \cap (D \cup V_j)}.
\end{align*}
% \begin{align*}
%         r^+_{(\mathcal{I}, (C, D))}(U_i) &= 
%                 \begin{cases}
%                     \max\set{\dist_G(U_i, x) \mid x\in X^+\cap B_G(U_i,r-\dist_{G[C \cup U_i]}(c,U_i))}\hspace{-100mm} \\
%                     \qquad &\textrm{if $X^+\cap  B_G(U_i,r-\dist_{G[C \cup U_i]}(c,U_i)) \neq \emptyset$,} \\[5pt]
%                     -1 & \textrm{otherwise,}
%                 \end{cases} \\
% \intertext{and}
%         r^-_{(\mathcal{I}, (C, D))}(U_i) &= 
%                 \begin{cases}
%                     \min\set{\dist_G(U_i,x)\mid x\in X^-} & \textrm{if $X^- \neq \emptyset$,} \\
%                     \infty & \textrm{otherwise.}
%                 \end{cases} \\[\bigskipamount]
% %\end{align*}
% \intertext{
% For every $j \in [\ell]$, we define
% }
% %\begin{align*}
%         r^+_{(\mathcal{I}, (C, D))}(V_j) &= 
%                 \begin{cases}
%                     \max\set{\dist_{G[D \cup V_j]}(V_j,x)\mid x\in X^+\cap B_{G[D \cup V_j]}(V_j,r-\dist_G(c,V_j))} \hspace{-100mm}\\
%                     \qquad &\textrm{if $X^+\cap B_{G[D \cup V_j]}(V_j,r-\dist_G(c,V_j)) \neq \emptyset$,} \\[5pt]
%                     -1 & \textrm{otherwise,}
%                 \end{cases} \\
% \intertext{and}
%         r^-_{(\mathcal{I}, (C, D))}(V_j) &= 
%                 \begin{cases}
%                     \min\set{\dist_{G[D \cup V_j]}(V_j,x)\mid x\in X^- \cap (D \cup V_j)} & \textrm{if $X^- \cap (D \cup V_j) \neq \emptyset$,} \\
%                     \infty & \textrm{otherwise.}
%                 \end{cases}
% \end{align*}

The next two statements: \Cref{lemma:good-balls-realize-again,lemma:initial-ball-realizes-again} are counterparts of \Cref{lemma:good-balls-realize,lemma:initial-ball-realizes} from the treewidth proof, respectively.
Namely, in~\cref{lemma:good-balls-realize-again}, we give a sufficient condition for balls to realize a sample and in~\cref{lemma:initial-ball-realizes-again}, we show that the initial ball of a sample compression scheme that realizes a given sample satisfies this condition.

\begin{lemma}\label{lemma:good-balls-realize-again}
    Let $\mathcal{I} = (G,(X^+,X^-),c,r)$ be an instance, let $X = X^+ \cup X^-$.
    Let $(C_1, D_1)$ and $(C_2, D_2)$ be ordered partitions of $V(G)$ such that $c \in C_1 \cap C_2$, and $X \subseteq D_1 \cup D_2$.
    Let $(U_1, \ldots, U_k)$ be a $C_1$-twin partition of $D_1$, and let $(V_1, \ldots, V_\ell)$ be a $D_2$-twin partition of $C_2$.
    % For every $i \in [k]$, let $r^+(U_i) = r^+_{(\mathcal{I}, (C_1, D_1))}(U_i)$ and $r^-(U_i) = r^-_{(\mathcal{I}, (C_1, D_1))}(U_i)$.
    % For every $j \in [\ell]$, let $r^+(V_j) = r^+_{(\mathcal{I}, (C_2, D_2))}(V_j)$ and $r^-(V_j) = r^-_{(\mathcal{I}, (C_2, D_2))}(V_j)$.
    Let $c' \in C_1 \cap C_2$ and let $s$ be an integer such that 
    \begin{align*}
        D_1 \cap \bigcup_{i \in [k]} B_G(U_i, r^+_{(\mathcal{I}, (C_1, D_1))}(U_i)) 
        &\subseteq %\\
        %&\hspace{-1.65cm} 
        D_1 \cap B_G(c', s)\\ 
        &\subseteq 
        D_1 \cap \bigcup_{i \in [k]} B_G(U_i, r^-_{(\mathcal{I}, (C_1, D_1))}(U_i) - 1), \\ 
    \intertext{and}
        D_2 \cap \bigcup_{j \in [\ell]} B_{G[D_2 \cup V_j]}(V_j, r^+_{(\mathcal{I}, (C_2, D_2))}(V_j)) 
        &\subseteq%\\
        %&\hspace{-1.65cm} 
        D_2 \cap B_{G}(c', s)\\ 
        &\subseteq D_2 \cap \bigcup_{j \in [\ell]} B_{G[D_2 \cup V_j]}(V_j, r^-_{(\mathcal{I}, (C_2, D_2))}(V_j) - 1). 
    \end{align*}
    Then, $B_G(c', s)$ realizes the sample $(X^+, X^-)$.
\end{lemma}

%\clement{A little clash of notation: it might be possible that $U_i = V_{i'}$ for some $i,i'$. However, $r^+(U_i)$ is a priori different from $r^+(V_{i'})$.}

\begin{proof}
It suffices to prove that $X^+ \subseteq B_G(c', s)$ and $X^- \cap B_G(c', s) = \emptyset$.
We start with the former.
Let $x \in X^+$. 
Since $B_G(c, r)$ realizes the sample $(X^+, X^-)$, we have $x \in B_G(c,r)$, and so, $\dist_G(c, x) \leq r$.
Since $x \in X \subseteq D_1 \cup D_2$, either $x \in D_1$ or $x \in D_2$.
Suppose first that $x \in D_1$.
By \cref{lemma:separator_with_few_types} applied to the ordered partition $(D_1, C_1)$ with $c_0=x$ and $d_0= c$ and the family $(U_1, \ldots, U_k)$, we have 
\[\dist_G(x, c) = \min_{i \in [k]}\left(\dist_G(x, U_i) + \dist_{G[C_1 \cup U_i]}(U_i, c)\right).\]
Thus, there exists $i \in [k]$ such that $\dist_G(x, U_i) + \dist_{G[C_1 \cup U_i]}(U_i, c) = \dist_G(x, c) \leq r$.
Then, $\dist_G(U_i, x) \leq r - \dist_{G[C_1 \cup U_i]}(U_i, c)$ so $x \in X^+ \cap B_G(U_i, r - \dist_{G[C_1 \cup U_i]}(U_i, c))$. 
Thus, by the definition of $r^+_{(\mathcal{I}, (C_1, D_1))}(U_i)$,
we have $\dist_G(U_i, x) \leq r^+_{(\mathcal{I}, (C_1, D_1))}(U_i)$ and so 
\[x \in D_1 \cap  B_G(U_i, r^+_{(\mathcal{I},(C_1,D_1))}(U_i)) \subset B_G(c',s).\]
% Therefore, 
% \[
%     x \in D_1 \cap  \bigcup_{i \in [k]}B_G(U_i, r^+_{(\mathcal{I},(C_1,D_1))}(U_i)) \subseteq B_G(c', s).
% \]
Now suppose that $x \in D_2$.
By \cref{lemma:separator_with_few_types} applied to the ordered partition $(C_2, D_2)$ with $c_0=c$ and $d_0 = x$ and the family $(V_1, \ldots, V_\ell)$, we have 
\[\dist_G(c, x) = \min_{j \in [\ell]}\left(\dist_G(c, V_j) + \dist_{G[D_2 \cup V_j]}(V_j, x)\right).\]

Thus, there exists $j \in [\ell]$ such that $\dist_G(c, V_j) + \dist_{G[D_2 \cup V_j]}(V_j, x) = \dist_G(c, x) \leq r$.
Then, $\dist_{G[D_2 \cup V_j]}(V_j, x) \leq r - \dist_G(c, V_j)$ so $x \in X^+ \cap B_{G[D_2 \cup V_j]}(V_j, r - \dist_G(V_j, c))$. Thus, by the definition of $r^+_{(\mathcal{I}, (C_2, D_2))}(V_j)$, we have $\dist_{G[D_2 \cup V_j]}(V_j, x) \leq r^+_{(\mathcal{I}, (C_2, D_2))}(V_j)$ and so
\[
    x \in D_2 \cap B_{G[D_2 \cup V_j]}(V_j, r^+_{(\mathcal{I}, (C_2, D_2))}(V_j)) \subseteq B_{G}(c', s).
\]
In both cases, we indeed have $x \in B_G(c', s)$, which concludes the proof that $X^+ \subset B_G(c',s)$.

Next, we prove that $X^- \cap B_G(c', s) = \emptyset$.
By contradiction, suppose that there exists $x \in X^- \cap B_G(c', s)$.
Since $x \in X \subseteq D_1 \cup D_2$, either $x \in D_1$ or $x \in D_2$.
Suppose first that $x \in D_1$.
Then, 
\[
    x \in D_1 \cap B_G(c', s)  \subseteq \bigcup_{i \in [k]}B_G(U_i, r^-_{(\mathcal{I}, (C_1, D_1))}(U_i)-1).
\]
Thus, there exists $i \in [k]$ such that $\dist_G(U_i, x) < r^-_{(\mathcal{I}, (C_1, D_1))}(U_i)$, which contradicts the definition of $r^-_{(\mathcal{I}, (C_1, D_1))}(U_i)$ as $x \in X^-$.
Finally, suppose that $x \in D_2$.
Then, 
\[
    x \in D_2 \cap B_G(c', s) \subseteq \bigcup_{j \in [\ell]}B_{G[D_2 \cup V_j]}(V_j, r^-_{(\mathcal{I}, (C_2, D_2))}(V_j) - 1).
\]
Thus, there exists $j \in [\ell]$ such that $\dist_{G[D_2 \cup V_j]}(V_j, x) < r^-_{(\mathcal{I}, (C_2, D_2))}(V_j)$, which contradicts the definition of $r^-_{(\mathcal{I}, (C_2, D_2))}(V_j)$ as $x \in X^- \cap D_2$.
Since in both cases, we obtained a contradiction, the proof is completed.
\end{proof}

% We now show that the sufficient condition in \Cref{lemma:good-balls-realize-again} on balls $B(c',s)$ to ensure that $B(c',s)$ realizes $(X^+,X^-)$ is actually satisfied by the original ball $B(c,r)$ realizing $(X^+,X^-)$.

\begin{lemma}\label{lemma:initial-ball-realizes-again}
    Let $\mathcal{I} = (G,(X^+,X^-),c,r)$ be an instance, let $X = X^+ \cup X^-$.
    Let $(C_1, D_1)$ and $(C_2, D_2)$ be ordered partitions of $V(G)$ such that $c \in C_1 \cap C_2$, and $X \subseteq D_1 \cup D_2$.
    Let $(U_1, \ldots, U_k)$ be a $C_1$-twin partition of $D_1$, and let $(V_1, \ldots, V_\ell)$ be a $D_2$-twin partition of $C_2$.
    Then,
    \begin{align*}
        D_1 \cap \bigcup_{i \in [k]} B_G(U_i, r^+_{(\mathcal{I}, (C_1, D_1))}(U_i)) 
        &\subseteq D_1 \cap B_G(c, r)\\ 
        &\subseteq D_1 \cap \bigcup_{i \in [k]} B_G(U_i, r^-_{(\mathcal{I}, (C_1, D_1))}(U_i) - 1), \\ 
    \intertext{and}
        D_2 \cap \bigcup_{j \in [\ell]} B_{G[D_2 \cup V_j]}(V_j, r^+_{(\mathcal{I}, (C_2, D_2))}(V_j))  
        &\subseteq D_2 \cap B_{G}(c, r)\\ 
        &\subseteq 
        D_2 \cap \bigcup_{j \in [\ell]} B_{G[D_2 \cup V_j]}(V_j, r^-_{(\mathcal{I}, (C_2, D_2))}(V_j) - 1). 
    \end{align*} 
\end{lemma}

% \begin{lemma}\label{lemma:initial-ball-realizes-again}
%     Let $\mathcal{I} = (G,(X^+,X^-),c,r)$ be an instance, let $X = X^+ \cup X^-$.
%     Let $(C_1, D_1)$ and $(C_2, D_2)$ be ordered partitions of $V(G)$ such that $c \in C_1 \cap C_2$, and $X \subseteq D_1 \cup D_2$.
%     Let $(U_1, \ldots, U_k)$ be a $C_1$-twin partition of $D_1$, and and let $(V_1, \ldots, V_\ell)$ be a $D_2$-twin partition of $C_2$.
%     For every $i \in [k]$, let $r^+(U_i) = r^+_{(\mathcal{I}, (C_1, D_1))}(U_i)$ and $r^-(U_i) = r^-_{(\mathcal{I}, (C_1, D_1))}(U_i)$.
%     For every $j \in [\ell]$, let $r^+(V_j) = r^+_{(\mathcal{I}, (C_2, D_2))}(V_j)$ and $r^-(V_j) = r^-_{(\mathcal{I}, (C_2, D_2))}(V_j)$.
%     Then,
%     \begin{align*}
%         \bigcup_{i \in [k]} \left(B_G(U_i, r^+(U_i)) \cap D_1 \right) 
%         &\subseteq B_G(c, r) \cap D_1 \subseteq 
%         \bigcup_{i \in [k]} \left(B_G(U_i, r^-(U_i) - 1) \cap D_1 \right), \\ 
%     \intertext{and}
%         \bigcup_{j \in [\ell]} \left(B_{G[D_2 \cup V_j]}(V_j, r^+(V_j)) \cap D_2 \right) 
%         &\subseteq B_{G}(c, r) \cap D_2 \subseteq 
%         \bigcup_{j \in [\ell]} \left(B_{G[D_2 \cup V_j]}(V_j, r^-(V_j) - 1) \cap D_2 \right). 
%     \end{align*}
% \end{lemma}

\begin{proof}
    There are four inclusions to prove, we begin with
    \[
        D_1 \cap \bigcup_{i \in [k]} B_G(U_i, r^+_{(\mathcal{I}, (C_1, D_1))}(U_i)) 
        \subseteq D_1 \cap B_G(c, r).
    \]
    Let $u$ be a vertex in the set on the left-hand side of the above inclusion.
    There exists $i \in [k]$ such that $u \in D_1 \cap B_G(U_i, r^+_{(\mathcal{I}, (C_1, D_1))}(U_i)) $.
    Then, $r^+_{(\mathcal{I}, (C_1, D_1))}(U_i) \neq -1$ so $r^+_{(\mathcal{I}, (C_1, D_1))}(U_i) \leq r-\dist_{G[C_1 \cup U_i]}(c,U_i)$.
    Therefore, 
    \[\dist_G(u, U_i) + \dist_{G[C_1 \cup U_i]}(U_i, c) \leq r^+_{(\mathcal{I}, (C_1, D_1))}(U_i) + \dist_{G[C_1 \cup U_i]}(U_i, c) \leq r.\]
    Thus, by \cref{lemma:separator_with_few_types} applied to the separation $(D_1, C_1)$ with $c_0 = u$, $d_0 = c$, and the family $(U_1, \ldots, U_k)$, we have $\dist_G(u, c) \leq \dist_G(u, U_i) + \dist_{G[C_1 \cup U_i]}(U_i, c) \leq r$.
    It follows that $u \in B_G(c, r)$.

    Next, we prove that 
    \[
        D_2 \cap \bigcup_{j \in [\ell]} B_{G[D_2 \cup V_j]}(V_j, r^+_{(\mathcal{I}, (C_2, D_2))}(V_j)) \subseteq D_2 \cap B_{G}(c, r).
    \]
    Let $u$ be a vertex in the set on the left-hand side of the above inclusion.
    There exists $j \in [\ell]$ such that $u \in B_{G[D_2 \cup V_j]}(V_j, r^+_{(\mathcal{I}, (C_2, D_2))}(V_j)) \cap D_2$.
    Then, $r^+_{(\mathcal{I}, (C_2, D_2))}(V_j) \neq -1$ so $r^+_{(\mathcal{I}, (C_2, D_2))}(V_j) \leq r-\dist_G(c,V_j)$.
    Therefore, 
    \[\dist_G(c, V_j) + \dist_{G[D_2 \cup V_j]}(V_j, u) \leq \dist_G(c, V_j) + r^+_{(\mathcal{I}, (C_2, D_2))}(V_j) \leq r.\]
    Thus, by \cref{lemma:separator_with_few_types} applied to the separation $(C_2, D_2)$ with $c_0 = c$ and $d_0 = u$, with the family $(V_1, \ldots, V_{\ell})$, we have $\dist_G(c, u) \leq \dist_G(c, V_j) + \dist_{G[D_2 \cup V_j]}(V_j, u) \leq r$.
    It follows that $v \in B_G(c, r)$.

    The next inclusion to prove is
    \[
        D_1 \cap B_G(c, r) \subseteq 
        D_1 \cap \bigcup_{i \in [k]} B_G(U_i, r^-_{(\mathcal{I}, (C_1, D_1))}(U_i) - 1).
    \]
    Let $u \in D_1 \cap B_G(c, r)$.
    By \cref{lemma:separator_with_few_types} applied to the separation $(D_1, C_1)$ with $c_0 = u$ and $d_0 = c$, and the family $(U_1, \ldots, U_k)$, there exists $i \in [k]$ such that $\dist_G(u, U_i) + \dist_{G[C_1 \cup U_i]}(U_i, c) = \dist_G(u, c) \leq r$.
    Suppose by contradiction that $u \notin B_G(U_i, r^-_{(\mathcal{I}, (C_1, D_1))}(U_i) - 1)$.
    Then, $r^-_{(\mathcal{I}, (C_1, D_1))}(U_i) \neq \infty$ so there exists $x \in X^-$ such that $\dist_G(U_i, x) = r^-_{(\mathcal{I}, (C_1, D_1))}(U_i) \leq \dist_G(U_i, u)$.
    Then, 
    \[\dist_G(U_i, x) + \dist_{G[C_1 \cup U_i]}(U_i, c) \leq \dist_G(u, U_i) + \dist_{G[C_1 \cup U_i]}(U_i, c) \leq r.\]
    By \cref{lemma:separator_with_few_types}, applied to the separation $(D_1, C_1)$ with $c_0 = x$ and $d_0 = c$, with the family $(U_1, \ldots, U_k)$, we get $\dist_G(x, c) \leq r$ so $x \in B_G(c, r) \cap X^-$, contradicting that $B_G(c, r)$ realizes $(X^+, X^-)$.
    This proves that 
    $u \in B_G(U_i, r^-_{(\mathcal{I}, (C_1, D_1))}(U_i) - 1)$, as desired.

    Finally, we prove that 
    \[
        D_2 \cap B_{G}(c, r) \subseteq 
        D_2 \cap \bigcup_{j \in [\ell]} B_{G[D_2 \cup V_j]}(V_j, r^-_{(\mathcal{I}, (C_2, D_2))}(V_j) - 1).
    \]
    Let $u \in D_2 \cap B_G(c, r)$.
    By \cref{lemma:separator_with_few_types} applied to the separation $(C_2, D_2)$ with $c_0 = c$ and $d_0 = u$, with the family $(V_1, \ldots, V_{\ell})$, there exists $j \in [\ell]$ such that $\dist_G(c, V_j) + \dist_{G[D_2 \cup V_j]}(V_j, u) = \dist_G(c, u) \leq r$.
    Suppose by contradiction that $u \notin B_{G[D_2 \cup V_j]}(V_j, r^-_{(\mathcal{I}, (C_2, D_2))}(V_j) - 1)$.
    Then, $r^-_{(\mathcal{I}, (C_2, D_2))}(V_j) \neq \infty$ so there exists $x \in X^-$ such that $\dist_{G[D_2 \cup V_j]}(V_j, x) = r^-_{(\mathcal{I}, (C_2, D_2))}(V_j) \leq \dist_{G[D_2 \cup V_j]}(V_j, u)$.
    Then, \[\dist_G(c, V_j) + \dist_{G[D_2 \cup V_j]}(V_j, x) \leq \dist_G(c, V_j) +\dist_{G[D_2 \cup V_j]}(V_j, u) \leq r.\]
    By \cref{lemma:separator_with_few_types}, applied to the separation $(C_2, D_2)$ with $c_0 = c$ and $d_0 = x$, with the family $(V_1, \ldots, V_{\ell})$, we get $\dist_G(c, x) \leq r$ so $x \in B_G(c, r) \cap X^-$, contradicting that $B_G(c, r)$ realizes $(X^+, X^-)$.
    This proves that $u \in B_{G[D_2 \cup V_j]}(V_j, r^-_{(\mathcal{I}, (C_2, D_2))}(V_j) - 1)$, as desired.
    This completes the proof of the last out of the four inclusions and so of the lemma.
\end{proof}

We previously mentioned that the graphs of bounded NLC-width can be decomposed along simple edge cuts. This is made precise by the following lemma.

\begin{lemma}\label{lem:NLC-few-nbhd-classes}
    Let $t$ be an integer, let $G$ be a graph and let $(T, Q, \alpha, \beta, R)$ be an NLC-decomposition of $G$ with $|Q| \leq t$.
    Let $y$ be a non-root node of $T$ and let $(C, D) = (V(T_{y \mid \parent(T, y)}) \cap V(G), V(T_{\parent(T, y) \mid y}) \cap V(G))$.
    For every $q \in Q$, let $U_q = \{u \in C \mid \beta(\{y, u\})(\alpha(u)) = q\}$.
    Then, the partition $(U_q)_{q \in Q}$ is a $D$-twin partition of $C$ (into $|Q| \leq t$ classes).
\end{lemma}

\begin{proof}
    Let $u, u' \in C$ be such that $\beta(\{y, u\})(\alpha(u)) = \beta(\{y, u'\})(\alpha(u'))$.
    Let $w \in D$ and let $z$ be the lowest common ancestor of $y$ and $w$.
    Since $u, u' \in C \subseteq V(T_{y \mid \parent(T, y)})$, we have that $u, u'$ are descendants of $y$.
    Since $w \in D \subseteq V(T_{\parent(T,y) \mid y})$, we have that $w$ is not a descendant of $y$.
    Therefore, $\lca(T, u, w) = z = \lca(T, u', w)$.
    
    Furthermore, $z$ is an ancestor of $y$ so using $\beta(\{y, u\})(\alpha(u)) = \beta(\{y, u'\})(\alpha(u'))$ and \ref{item:nlc1}, we obtain 
    \[\beta(\{z, u\})(\alpha(u)) = \beta(\{z, y\}) \circ \beta(\{y, u\})(\alpha(u)) = \beta(\{z, y\}) \circ \beta(\{y, u'\})(\alpha(u')) = \beta(\{y, u'\})(\alpha(u')).\]
    To conclude, note that \ref{item:nlc2} implies that $uw \in E(G)$ if and only if $u'w \in E(G)$.
\end{proof}

Next, we prove the main result of this section, \cref{thm:cw}, which we restate for convenience.

\thmcw*

\begin{proof}
    Let $t$ be a positive integer and let $G$ be a graph of cliquewidth (and so NLC-width) at most $t$.
    Let $H$ be the hypergraph of balls in $G$.
    We fix an NLC-decomposition $(T, Q, \alpha, \beta, R)$ of $G$ with $Q = [t]$.
    We also fix a map $\varphi\colon (V(T) \cup \{\emptyentry\})^3 \to S(T)$ which satisfies the condition of \cref{lemma:encode-subtree}. 

    \header{Compression.}
            Let $(X^+,X^-)$ be a sample of $H$,
            and let $X = X^+ \cup X^-$.            
            We fix a ball $B_G(c, r)$ realizing $(X^+,X^-)$, 
            and we consider the instance $\mathcal{I} = (G,(X^+,X^-),c,r)$.
            We will construct $\kappa(X^+,X^-) \in (V(H)\cup \{\emptyentry\})^{4t+3}$.
            First, we consider several degenerate cases.
            If $X = \emptyset$, then we set $\kappa(X^+, X^-) = (\emptyentry, \emptyentry, \ldots, \emptyentry)$.
            If $X^+ \neq \emptyset$ and $X^- = \emptyset$, we set $\kappa(X^+, X^-) = (\emptyentry, \emptyentry, \ldots, \emptyentry,x^+,x^+)$ where $x^+ \in X^+$ is chosen arbitrarily. 
            From now on, we assume that $X^- \neq \emptyset$.
            If $c \in X$, let $x^- \in X^-$ be a vertex of $X^-$ which minimizes $\dist_G(c, x^-)$.
            In that case, set $\kappa(X^+, X^-) = (c, \emptyentry, \emptyentry, x^-, x^-, \emptyentry, \emptyentry, \ldots, \emptyentry)$.
            %Observe that in that case the fourth and the fifth entries of $\kappa(X^+, X^-)$ are equal, and different from $\emptyentry$.
            From now on, we assume that $c \notin X$.
            Recall that $V(G)$ is the set of all the leaves of $T$,
            and so $X$ is a subset of leaves of $T$, and $c \not\in \LCA(T,X)$. 
            Let $C$ be the component of $T - \LCA(T, X)$ which contains $c$.
            There exist $w_1, w_2, w_3 \in X \cup \{\emptyentry\}$ such that $\varphi(w_1, w_2, w_3) = C$.
            Thus, $\varphi(w_1, w_2, w_3)$ is the component of $T - \LCA(T, X)$ which contains $c$.

            By \cref{lemma:increase_Z_in_a_tree}, there are at most two edges $e_1, e_2 \in E(T)$ connecting $V(C)$ to $V(T)\setminus V(C)$ in $T$. If there is only one such edge (or zero), we denote the other one by $\emptyentry$.
            Up to renaming $e_1$ and $e_2$, we then either have $e_1 = \emptyentry$ or there exists $y_1 \in V(T)$ such that $e_1 = \{y_1, \parent(T, y_1)\}$ and $V(C) \subseteq V(T_{\parent(T, y_1) \mid y_1})$, and either $e_2 = \emptyentry$ or there exists $y_2 \in V(T)$ such that $e_2 = \{y_2, \parent(T, y_2)\}$ and $V(C) \subseteq V(T_{y_2 \mid \parent(T, y_2)})$.

            If $e_1 = \emptyentry$, let $(C_1, D_1) = (V(G), \emptyset)$.
            Otherwise, let $(C_1, D_1) = (V(T_{\parent(T, y_1) \mid y_1}) \cap V(G), V(T_{y_1 \mid \parent(T, y_1)}) \cap V(G))$.
            If $e_2 = \emptyentry$, let $(C_2, D_2) = (V(G), \emptyset)$.
            Otherwise, let $(C_2, D_2) = (V(T_{y_2 \mid \parent(T, y_2)}) \cap V(G), V(T_{\parent(T, y_2) \mid y_2}) \cap V(G))$.
            Observe that $c \in V(C) \subseteq  C_1 \cap C_2$ by definition of $(C_1, D_1)$ and $(C_2, D_2)$.

            If $e_1 \neq \emptyentry$, then for every $i \in [t]$, let $U_i = \{u \in D_1 \mid \beta(\{y_1, u\})(\alpha(u)) = i\}$. 
            Otherwise, set $U_i = \emptyset$ for every $i \in [t]$.
            If $e_2 \neq \emptyentry$, for every $j \in [t]$, let $V_j = \{u \in C_2 \mid \beta(\{y_2, u\})(\alpha(u)) = j\}$. 
            Otherwise, set $V_j = \emptyset$ for every $j \in [t]$.

            For all $i,j \in [t]$, we define $x_i^+,x_{t+j}^+ \in X^+\cup \{\emptyentry\}$ and $x_i^-, x_{t+j}^- \in X^- \cup \{\emptyentry\}$ so that if possible, we choose respective vertices of $X^+$ or $X^-$ such that
            \begin{align*}
                \dist_G(U_i, x_i^+) &= r^+_{(\mathcal{I}, (C_1, D_1))}(U_i) \text{, } \ \  \dist_{G[D_2 \cup V_j]}(V_j, x^+_{t+j}) = r^+_{(\mathcal{I}, (C_2, D_2))}(V_j) \text{; and }\\ \dist_G(U_i, x^-_{j}) &= r^-_{(\mathcal{I}, (C_1, D_1))}(U_i) \text{, } \ \ \dist_{G[D_2 \cup V_j]}(V_j, x^-_{t+j}) = r^-_{(\mathcal{I}, (C_2, D_2))}(V_j).
            \end{align*}
            Otherwise, we set the respective vertices to $\emptyentry$.
            % For every $i \in [t]$, let $x_i^+ \in X^+$ be a vertex such that $r^+_{(\mathcal{I}, (C_1, D_1))}(U_i) = \dist_G(U_i, x_i^+)$ if such a vertex exists, and set $x_i^+ = \emptyentry$ otherwise.
            % Let also $x_i^- \in X$ be a vertex such that $r^-_{(\mathcal{I}, (C_1, D_1))}(U_i) = \dist_G(U_i, x)$ if such a vertex exists, and set $x_i^- = \emptyentry$ otherwise.
            % For every $j \in [t]$, let $x_{t+j}^+ \in X$ be a vertex such that $\dist_{G[D_2 \cup V_j]}(V_j, x) = r^+_{(\mathcal{I}, (C_2, D_2))}(V_j)$ if such a vertex exists, and set $x_{t+j}^+ = \emptyentry$ otherwise.
            % Let also $x_{t+j}^- \in X$ be a vertex such that $\dist_{G[D_2 \cup V_j]}(V_j, x) = r^-_{(\mathcal{I}, (C_2, D_2))}(V_j)$ if such a vertex exists, and set $x_{t+j}^- = \emptyentry$ otherwise.
            Finally, we set 
            \[\kappa(X^+, X^-) = (w_1, w_2, w_3, x_1^+, x_1^-, \ldots, x_{2t}^+, x_{2t}^-).\]
            Observe that in this case, the fourth and fifth entries of $\kappa(X^+, X^-)$ can not be equal and different from $\emptyentry$ (since one is either $\emptyentry$ or in $X^+$ and the other is either $\emptyentry$ or in $X^-$).
            The same holds for the last and penultimate entries.
            
        \header{Reconstruction.}
            Consider an input vector for the reconstructor $\rho$:
            \[
                (w_1, w_2, w_3, x_1^+, x_1^-, \ldots, x_{2t}^+, x_{2t}^-) \in (V(G) \cup \{\emptyentry\})^{4t+3}.
            \]

            If $(w_1, w_2, w_3, x_1^+, x_1^-, \ldots, x_{2t}^+, x_{2t}^-) = (\emptyentry, \emptyentry, \ldots, \emptyentry)$, then $\rho$ returns any ball.
            If $x^+_{2t} = x^-_{2t}$ and $ x^+_{2t} \neq \emptyentry$, then $\rho$ returns $B_G(x_{2t}^+, |V(G)|)$.
            If $x^+_1 = x^-_1$ and $ x^+_1 \neq \emptyentry$, then $\rho$ returns $B_G(w_1, \dist_G(w_1, x^-_1) - 1)$.

            Assume that none of the above cases hold.
            Let $\varphi(w_1,w_2,w_3) = C$.
            We replicate the definitions of $e_1,e_2, y_1,y_2,(C_1,D_1),(C_2,D_2),U_1,\dots,U_t,V_1,\dots,V_t$ for $C$ after $\kappa$.
            Note that these objects depend only on the given NLC-decomposition and $C$, hence, this is well-defined.

            We additionally define, for all $i, j \in [t]$:
            \begin{align*}
                r^+_1(U_i) &= \begin{cases}
                    \dist_G(U_i,x_i^+) &\text{if $x_i^+ \in V(G)$,}\\
                    -1 &\text{if $x_i^+ = \emptyentry$;}
                \end{cases}
                &&r^+_2(V_j) = \begin{cases}
                    \dist_{G[D_2 \cup V_j]}(V_j,x_{t+j}^+) &\text{if $x_{t+j}^+ \in V(G)$,}\\
                    -1 &\text{if $x_{t+j}^+ = \emptyentry$;}
                \end{cases}\\
                r^-_1(U_i) &= \begin{cases}
                    \dist_G(U_i,x_i^-) &\text{if $x_i^- \in V(G)$,}\\
                    \infty &\text{if $x_i^- = \emptyentry$;}
                \end{cases}
                &&r^-_2(V_j) = \begin{cases}
                    \dist_{G[D_2 \cup V_j]}(V_j,x_{t+j}^-) &\text{if $x_{t+j}^- \in V(G)$,}\\
                    \infty &\text{if $x_{t+j}^- = \emptyentry$.}
                \end{cases}
            \end{align*}
            % For every $i \in [t]$, set $r^+_1(U_i) = \dist_G(U_i, x_i^+)$ if $x_i^+ \neq \emptyentry$ and set $r^+_1(U_i) = -1$ otherwise.
            % Set also $r^-_1(U_i) = \dist_G(U_i, x_i^-)$ if $x_i^- \neq \emptyentry$ and set $r^-_1(U_i) = \infty$ otherwise.
            % For every $j \in [t]$, set $r^+_2(V_j) = \dist_{G[D_2 \cup V_j]}(V_j, x_{t+j}^+)$ if $x_{t+j}^+ \neq \emptyentry$ and set $r^+_2(V_j) = -1$ otherwise.
            % Set also $r^-_2(V_j) = \dist_{G[D_2 \cup V_j]}(V_j, x_{t+j}^-)$ if $x_{t+j}^- \neq \emptyentry$ and set $r^-_2(V_j) = \infty$ otherwise.
            
            To decide the output, $\rho$ searches for a ball $B_G(c', s)$ with $c' \in C_1 \cap C_2$ such that
            \begin{align*}
                D_1 \cap \bigcup_{i \in [t]} B_G(U_i, r^+_1(U_i)) 
                &\subseteq D_1 \cap B_G(c', s) \subseteq 
                D_1 \cap \bigcup_{i \in [t]} B_G(U_i, r^-_1(U_i) - 1)
            \intertext{and}
                D_2 \cap \bigcup_{j \in [t]} B_{G[D_2 \cup V_j]}(V_j, r^+_2(V_j))
                &\subseteq D_2 \cap B_{G}(c', s) \subseteq 
                D_2 \cap \bigcup_{j \in [t]} B_{G[D_2 \cup V_j]}(V_j, r^-_2(V_j) - 1). 
            \end{align*}
            If such a ball $B_G(c',s)$ exists, $\rho$ returns it, and otherwise $\rho$ returns an arbitrary ball in $G$.
 
    \header{Proof of correctness.}
    We now prove that $(\kappa,\rho)$ 
    is a proper array sample compression scheme of $H$. 
    By definition, $\rho$ only outputs balls in $G$.
    Clearly, for each sample $(X^+, X^-)$ of $H$, 
    we get $\kappa(X^+,X^-)$ which is an element of $(X\cup\set{\emptyentry})^{4t+3}$. 
    We need to argue that $\rho(\kappa(X^+, X^-))$ realizes $(X^+, X^-)$.

    Let $(X^+,X^-)$ be a sample of $H$ and let
    \[\kappa(X^+, X^-) = (w_1, w_2, w_3, x_1^+, x_1^-, \ldots, x_{2t}^+, x_{2t}^-).\]
    Let $C = \varphi(w_1, w_2, w_3)$. 
    Let $c\in V(G)$ and an integer $r$ be as they were fixed in the definition of $\kappa(X^+, X^-)$. 
    Thus, $B_G(c, r)$ realizes $(X^+, X^-)$.
    Note that $\mathcal{I} = (G,(X^+,X^-),c,r)$ is an instance.    

    If $X = \emptyset$ then $\kappa(X^+, X^-) = (\emptyentry, \emptyentry, \ldots, \emptyentry)$ so the reconstructor outputs any ball, which indeed realizes $(X^+, X^-)$.
    If $X^+ \neq \emptyset$ and $X^- = \emptyset$, then there is $x^+ \in X^+$ with $x^+ = x^+_{2t} = x^-_{2t}$.
    In this case $\rho$ returns $B_G(x^+, |V(G)|)$, which is equal to the vertex set of the component of $G$ containing $X^+$.
    Such a ball realizes $(X^+, X^-)$.
    From now on, we assume that $X^- \neq \emptyset$.
    If $c \in X$ then $w_1 = c$ and $x^+_1 = x^-_1 = x^-$ where $x^- \in X^-$ minimizes $\dist_G(c, x^-)$. 
    In this case $\rho$ returns $B_G(c,\dist_G(c,x^-)-1) = B_G(c,\dist_G(c,X^-)-1)$.
    Clearly, $X^- \cap B_G(c,\dist_G(c,X^-)-1) = \emptyset$, and $X^+ \subset B_G(c,r) \subset B_G(c,\dist_G(c,X^-)-1)$.
    In particular, the returned ball realizes $(X^+, X^-)$.
    From now on, we assume that $c \notin X$.

    The compressor $\kappa$ chose $w_1$, $w_2$, and $w_3$, so that $C$ is the component of $T - \LCA(T, X)$ which contains $c$.
    By construction, the compressor $\kappa$ and the reconstructor $\rho$ consider the same ordered partitions $(C_1, D_1)$, $(C_2, D_2)$, $(U_1, \ldots, U_{t})$, and $(V_1, \ldots, V_{t})$. 
    By \cref{lem:NLC-few-nbhd-classes}, $(U_1, \ldots, U_{t})$ is a $C_1$-twin partition of $D_1$ and $(V_1, \ldots, V_{t})$ is a $D_2$-twin partition of $C_2$.

    It follows from construction that for all $i,j \in [t]$, we have
    \begin{align*}
        r^+_1(U_i) = r^+_{(\mathcal{I}, (C_1, D_1))}(U_i), \ \ &r^+_2(V_j) = r^+_{(\mathcal{I}, (C_2, D_2))}(V_j),\\
        r^-_1(U_i) = r^-_{(\mathcal{I}, (C_1, D_1))}(U_i), \ \ &r^-_2(V_j) = r^-_{(\mathcal{I}, (C_2, D_2))}(V_j).
    \end{align*}

    We claim that the partitions $(C_1, D_1)$ and $(C_2, D_2)$ satisfy $c \in C_1 \cap C_2$ and $X \subseteq D_1 \cup D_2$. 
    We already observed that $c \in C_1 \cap C_2$ in the compression phase.
    Furthermore, since $c \notin X$ then $C$ is the component of $c$ in $T - \LCA(T, X)$ so $V(C)$ is disjoint from $X$. 
    Additionally, $C_1 \cap C_2 = V(C) \cap V(G)$. 
    Therefore, $C_1 \cap C_2$ is disjoint from $X$, but $(C_1 \cap C_2, D_1 \cup D_2)$ is a partition of $V(G)$ so $X \subseteq D_1 \cup D_2$.

    By \cref{lemma:initial-ball-realizes-again}, it follows that
    \begin{align*}
        D_1 \cap \bigcup_{i \in [t]} B_G(U_i, r^+_1(U_i))
        &\subseteq D_1 \cap B_G(c, r) \subseteq 
        D_1 \cap \bigcup_{i \in [t]} B_G(U_i, r^-_1(U_i) - 1)
    \intertext{and}
        D_2 \cap \bigcup_{j \in [t]} B_{G[D_2 \cup V_j]}(V_j, r^+_2(V_j))
        &\subseteq D_2 \cap B_{G}(c, r) \subseteq 
        D_2 \cap \bigcup_{j \in [t]} B_{G[D_2 \cup V_j]}(V_j, r^-_2(V_j) - 1). 
    \end{align*}
    Thus, the search of the reconstructor $\rho$ for a ball satisfying the inclusions is successful, and in particular $\rho$ returns any ball $B_G(c', s)$ with $c' \in C_1 \cap C_2$ such that 
    \begin{align*}
        D_1 \cap \bigcup_{i \in [t]} B_G(U_i, r^+_1(U_i)) 
        &\subseteq D_1 \cap  B_G(c', s) \subseteq 
        D_1 \cap \bigcup_{i \in [t]} B_G(U_i, r^-_1(U_i) - 1) 
    \intertext{and}
        D_2 \cap \bigcup_{j \in [t]} B_{G[D_2 \cup V_j]}(V_j, r^+_2(V_j))
        &\subseteq D_2 \cap B_{G}(c', s) \subseteq 
        D_2 \cap \bigcup_{j \in [t]} B_{G[D_2 \cup V_j]}(V_j, r^-_2(V_j) - 1). 
    \end{align*}
    By \cref{lemma:good-balls-realize-again}, we deduce that $B_G(c',s) = \rho(\kappa(X^+,X^-))$ realizes $(X^+, X^-)$.
\end{proof}

\section{Bounded vertex cover}\label{sec:vc}

In this section, we give a proper sample compression scheme of size $t+4$ for the hypergraphs of balls in graphs admitting a vertex cover of size at most $t$.
Note that the main technical difficulty comes from compressing balls of radius $1$.

\thmvc*

\begin{proof}
    Let $t$ be a positive integer, let $G$ be a graph containing a vertex cover $R$ of size $t$, and let $H$ be the hypergraph of balls in $G$.
    Let $R = \{r_1,\dots,r_t\}$.
    We fix an arbitrary total order $\prec$ on $V(G)$.
    
    \header{Compression.}
            Let $(X^+,X^-)$ be a sample of $H$ and fix $c\in V(G)$ and an integer $r$ such that $B_G(c, r) \cap (X^+ \cup X^-) = X^+$.
            Among all such pairs $(c, r)$, pick one where $r$ is maximum.
            Consider the following four mutually disjoint cases.%, where in each case we assume that the previous ones do not hold.
            %\begin{itemize}
            \begin{enumorig}[label={Case \arabic*:}, ref={\arabic*}, left=20pt]
                \item  $X^+ = \emptyset$. \label{case:1}
                \item $c \in X^+$. \label{case:2}
                \item $X^+ \neq \emptyset$, $c \notin X^+$, and $c \in R$. \label{case:3}
                \item $X^+ \neq \emptyset$, $c \notin X^+$, and $c \notin R$. \label{case:4}
            \end{enumorig}
            %\end{itemize}
            Let $b \in \{0,1\}^2$ be a bitstring encoding which case out of the four holds.
            In the case where $X^- \neq \emptyset$ we define as follows a special element $x \in X^-$.
            Let $S=\set{y\mid y \in X^-,\ y \notin R,\ N_G(y) = N_G(c)}$. % be the set of all elements $y \in X^-$ with $y \notin R$ such that $N_G(c) = N_G(y)$.
            The ordering $\prec$ induces a cyclic ordering on $S \cup \{c\}$.
            If $r = 1$, $c \notin R$ and $S \neq \emptyset$, we set $x$ to be the element of $S$ which immediately precedes $c$ in this cyclic ordering. 
            Otherwise, let $x$ be any vertex in $X^-$ minimizing $\dist_G(c,x)$.
            Additionally, when Case~\ref{case:3}~or~\ref{case:4} holds, we set $z$ to be any vertex of $X^+$.
            Depending on the four cases, we define a subsample $(Y^+, Y^-)$ of $(X^+,X^-)$. Suppose that case $i \in [4]$ holds. Let
            \[
                (Y^+,Y^-) = 
                \begin{cases}
                    (\emptyset,\emptyset) & \text{if $i = 1$,}\\
                    (\{c\},\{x\}) & \text{if $i = 2$ and $X^- \neq \emptyset$,}\\
                    (\{c\},\emptyset) & \text{if $i = 2$ and $X^- = \emptyset$,}\\
                    (\{z\},\{x\}) & \text{if $i \in \{3,4\}$ and $X^- \neq \emptyset$,}\\
                    (\{z\},\emptyset) & \text{if $i \in \{3,4\}$ and $X^- = \emptyset$.}
                \end{cases}
            \]
            Next, we define another bitstring $b' \in \{0,1\}^t$ again depending on which of the four cases holds.
            If Case~\ref{case:1} or~\ref{case:2} holds, then we set $b'$ arbitrarily.
            If Case~\ref{case:3} holds, then $c = r_j$ for some $j \in [t]$, and we set $b'$ to be all zeros except $1$ on the $j$th coordinate.
            If Case~\ref{case:4} holds, then for each $j \in [t]$, we set $b'_j = 1$ if and only if $r_j \in N_G(c)$.
            Finally, set
            \[
                \kappa(X^+, X^-) = ((Y^+,Y^-), bb').
            \]     
            
        \header{Reconstruction.}
            Consider an input for the reconstructor $\rho$:
            \[((Y^+,Y^-), bb') \in \calS(H) \times \{0,1\}^{t+2}.
            \]
            We assume that the input is of the form encoded by $\kappa$; otherwise it does not matter what the reconstructor returns.
            First, $b$ tells us in which case we are.
            Also, let $c'$ be the element of $Y^+$ whenever $Y^+ \neq \emptyset$ and let $x$ be the element of $Y^-$ whenever $Y^- \neq \emptyset$.
            If Case~\ref{case:1} holds, then $\rho$ returns the empty ball in $G$. 
            From now on, assume that Case~\ref{case:1} does not hold.
            Then, $Y^+ \neq \emptyset$ so $c'$ is well-defined.
            If $Y^- = \emptyset$ then $\rho$ returns the connected component of $c'$ in $G$ (which is indeed a ball in $G$).
            % \clement{What if $G$ is not connected ? I guess we need to remember one vertex in $X^+$ to know in which component we are.}
            % \jedrzej{Nice catch, fixed}
            From now on, we assume that $Y^- \neq \emptyset$, so $x$ is well-defined.
            If Case~\ref{case:2} holds, then $\rho$ returns $B_G(c',\dist_G(c',x)-1)$.
            If Case~\ref{case:3} holds, then let $j\in [t]$ be given by $b'$. Then, $\rho$ returns $B_G(r_j,\dist_G(r_j,x)-1)$.
            Finally, assume that Case~\ref{case:4} holds and let $R' \subset R$ be given by $b'$.
            Let $S$ be the set of all $y \in V(G) \setminus R$ such that $N_G(y) = R'$.
            First, if $x \notin S$, then we pick $y \in S$ arbitrarily and $\rho$ returns $B_G(y,\dist_G(y,x)-1)$.
            If $x \in S$, then we pick $y$ as the element of $S$ which immediately follows $x$ in the cyclic ordering of $S\cup\set{x}$ induced by $\prec$, and $\rho$ returns $B_G(y,1)$. 
    
    \header{Proof of correctness.}
    We now prove that $(\kappa,\rho)$ 
    is a proper sample compression scheme of $H$.
    By construction, $\kappa$ compresses every sample to a subsample of at most two elements plus $t+2$ additional bits, so the size of the scheme is $t+4$.
    By construction, $\rho$ only outputs balls in $G$, so the scheme is proper.
    Let $(X^+,X^-)$ be a sample of $H$.
    We now argue that $\rho(\kappa(X^+, X^-))$ realizes $(X^+, X^-)$.
    Let $\kappa(X^+, X^-) = ((Y^+,Y^-), bb')$, and fix the same $c\in V(G)$ and integer $r$ as $\kappa$ did, so that $B_G(c, r) \cap (X^+ \cup X^-) = X^+$.

    Because of the bitstring $b$, the compressor and the reconstructor consider the same case.
    We split the reasoning according to which case they consider.
    Let $c'$ be the element of $Y^+$ whenever $Y^+ \neq \emptyset$ and let $x$ be the element of $Y^-$ whenever $Y^- \neq \emptyset$.
    
    If Case~\ref{case:1} holds, i.e.\ $X^+ = \emptyset$, then $\rho$ returns the empty ball, which realizes $(\emptyset,X^-)$.

    We can now assume that Case~\ref{case:1} does not hold. Therefore, $X^+ \neq \emptyset$, and so also $Y^+ \neq \emptyset$, so $c'$ is well-defined.
    If $Y^- = \emptyset$, then $X^- = \emptyset$ by definition of $\kappa$, and $\rho$ returns a ball containing all the vertices of the component of $c'$ in $G$, hence all vertices in $X^+$.
    Such a ball realizes $(X^+,\emptyset)$.
    From now on, we assume that $X^+ \neq \emptyset$ and $Y^- \neq \emptyset$, so both $c'$ and $x$ are well-defined.
    
    If Case~\ref{case:2} holds, then $\rho$ returns $B_G(c',\dist_G(c',x)-1)$.
    We have $c = c'$ by definition of $\kappa$ in that case, and $x \in X^-$ so $x \notin B_G(c, r)$ so $r+1 \leq \dist_G(c',x) \leq \dist_G(c',x')$ for every $x' \in X^-$, because of how $\kappa$ chose $x$. Thus, $B_G(c',\dist_G(c',x)-1)$ realizes $(X^+, X^-)$.

    If Case~\ref{case:3} holds, then $\rho$ returns $B_G(r_j,\dist_G(r_j,x)-1)$ where $j\in [t]$ is given by $b'$.
    Similarly, we have $c = r_j$ and $r+1 \leq \dist_G(r_j,x) \leq \dist_G(r_j,x')$ for every $x' \in X^-$, hence, $B_G(r_j,\dist_G(r_j,x)-1)$ realizes $(X^+, X^-)$.

    Finally, we assume that Case~\ref{case:4} holds.
    Let $R' \subset R$ be given by $b'$ and let $S=\set{y\mid y \in V(G) \setminus R,\ N_G(y) = R'}$. % be the set of all $y \in V(G) \setminus R$ such that $N_G(y) = R'$. 
    Observe that $c \in S$ by the definition of $\kappa$ in that case.
    Let $y \in S$ be the center of the ball returned by $\rho$.
    %Here, $\rho$ returns $B_G(y,\dist_G(y,x)-1)$ for some $y \in S$, fix this $y$.
    Note that if $r \geq 2$, then $B_G(y,r) = B_G(c,r)$, thus $B_G(y,\dist_G(y,x)-1)$ realizes $(X^+, X^-)$, as before.
    Note that if $r \leq 0$, then either Case~\ref{case:1}~or~\ref{case:2} holds, hence, we may assume that $r = 1$.
    In this case, $\rho$ returns either $B_G(y,1)$ or $B_G(y,\dist_G(y,x)-1)$.
    However, in the latter case, by the maximality of $r$, we have $\dist_G(y,x) = 2$,
    and so the two cases are identical.
    Observe that $N_G(y) = N_G(c) = R'$.
    Also, $c \notin X^+$ as we are in Case~\ref{case:4}.
    In particular, $B_G(y,1)$ realizes $(X^+, X^-)$ as long as $y \notin X^-$.
    Suppose to the contrary that $y \in X^-$.
    Then, $y \neq c$ so we are in the special case where $r=1$, $c \notin R$ and $S \neq \emptyset$.
    In this case, by definition of $\kappa$, $x$ is the element of $S \cap X^-$ which immediately precedes $c$ in the cyclic ordering of $S\cup\set{x}$ induced by $\prec$. Then, $\rho$ defined $y$ as the element of $S$ which follows of $x$ in this cyclic order, and so either $y = c$ or $y \notin X^-$. 
    In both cases, we have $y \notin X^-$.
    This proves that $B_G(y,1)$ realizes $(X^+,X^-)$,
    which completes the proof.
\end{proof}

\section{Hypergraphs of balls of bounded radius}\label{sec:bd-radius}

In this section, we consider hypergraphs of balls of bounded radius.

\subsection{Local treewidth}\label{sec:local-treewidth}

Grohe~\cite{G03} introduced the following graph parameter.
Given a function $f\colon \N \rightarrow \N$, a graph $G$ has \defin{local treewidth} at most $f$, if for every positive integer $r$ and for every $v \in V(G)$, the treewidth of $G[B_G(v,r)]$ is at most $f(r)$.

\begin{theorem}\label{thm:local-tw}
    For every function $f\colon \N \rightarrow \N$, for every graph $G$ of local treewidth at most $f$, and for every positive integer $r$, the hypergraph of balls of radius at most $r$ in $G$ has a proper array sample compression scheme of size $4f(2r)+8$. 
\end{theorem}
\begin{proof}
    Let $f\colon \N \rightarrow \N$ be a function, let $G$ be a graph of local treewidth at most $f$, let $r$ be a positive integer, and let $H$ be the hypergraph of balls of radius at most $r$ in $G$.
    For every $v \in V(G)$, let $H_v$ be the hypergraph of balls of radius at most $r$ in $G[B_G(v,2r)]$.
    By the assumption on the local treewidth of $G$, applying~\Cref{cor:tw}, we may, for every $v \in V(G)$, fix a proper array sample compression scheme $(\kappa_v,\rho_v)$ for $H_v$ of size $4f(2r) + 7$.
    Let $\ell = 4f(2r) + 7$.

    \header{Compression.}
            Let $(X^+,X^-)$ be a sample of $H$ and fix $c \in V(G)$ and an integer $s$ such that $B_G(c,s) \cap (X^+ \cup X^-) = X^+$.
            We define a special element $x \in V(G) \cup \{\emptyentry\}$.
            If $X^+ = \emptyset$, then $x = \emptyentry$, and otherwise, we pick as $x$ any member of $X^+$.
            In the latter case, let $\kappa_x(X^+ \cap B_G(x,2r),X^- \cap B_G(x,2r)) = (w_1,\dots,w_\ell)$.
            Finally, set
            \[
                \kappa(X^+, X^-) = 
                \begin{cases}
                    (x,w_1,\dots,w_\ell) &\text{if $X^+ \neq \emptyset$,}\\
                    (\emptyentry,\dots,\emptyentry) &\text{if $X^+ = \emptyset$.}
                \end{cases}
            \]
    
    \header{Reconstruction.}
            Consider an input for the reconstructor $\rho$:
            \[(x,w_1,\dots,w_\ell) \in (V(G) \cup \{\emptyentry\})^{\ell+1}.
            \]
            We assume that the input is of the form encoded by $\kappa$; otherwise it does not matter what the reconstructor $\rho$ returns.
            If $x = \emptyentry$, then $\rho$ returns the empty ball.
            Otherwise, $\rho$ returns the ball $\rho_x(w_1,\dots,w_\ell)$.

    \header{Proof of correctness.}
        We now prove that $(\kappa,\rho)$ 
        is a proper array sample compression scheme of $H$. 
        Let $(X^+,X^-)$ be a sample of $H$.
        Clearly $\kappa(X^+,X^-) \in (X^+ \cup X^-\cup\set{\emptyentry})^{\ell+1}$. 
        We need to argue that $\rho(\kappa(X^+, X^-))$ is a ball of radius at most $r$ in $G$ that realizes $(X^+, X^-)$.
        Let $\kappa(X^+,X^-) = (x,w_1,\dots,w_\ell)$, and fix the same $c \in V(G)$ and integer $s$ as $\kappa$ did, so that $B_G(c,s) \cap (X^+ \cup X^-) = X^+$.

        If $X^+ = \emptyset$, then $x = \emptyentry$ and $\rho$ returns the empty ball, which realizes $(\emptyset,X^-)$.
        Suppose now that $X^+ \neq \emptyset$, so $x \in X^+$ by definition of $\kappa$.

        Note that for each $c' \in V(G)$ and for each integer $s'$ with $s' \leq r$, if $x \in B_G(c',s')$, then $B_G(c',s') \subset B_G(x,2r)$.
        Let $G_x = G[B_G(x,2r)]$.
        Recall that $(w_1,\dots,w_\ell) = \kappa_x(X^+ \cap B_G(x,2r),X^- \cap B_G(x,2r))$ and that $\rho(x,w_1,\dots,w_\ell) = \rho_x(w_1,\dots,w_\ell)$ is a ball $B_{G_x}(c',s')$ for some $c' \in V(G_x)$ and some integer $s'$ realizing the sample $(X^+ \cap B_G(x,2r),X^- \cap B_G(x,2r))$ of $H_x$.
        Since $(\kappa_x,\rho_x)$ comes from~\Cref{cor:tw}, we know that $s' \leq r$.
        Additionally, $x \in X^+ \cap B_G(x,2r)$, hence, $x \in B_{G_x}(c', s')$, so $B_{G_x}(c',s') = B_G(c',s')$.
        Since $X^+ \subset B_G(c,r)$, we in fact have $X^+ \cap B_G(x,2r) = X^+$.
        It follows that $X^+ \subset B_G(c',s')$.
        On the other hand, $B_G(c',s')$ is disjoint from $X^- \cap B_G(x,2r)$ and since $B_G(c',s') \subset B_G(x,2r)$, the ball $B_G(c',s')$ is also disjoint from $X^- \setminus B_G(x,2r)$.
        Altogether, we obtain that $\rho(\kappa(X^+, X^-))$ is a ball in $G$ that realizes $(X^+, X^-)$.
        This completes the proof.
\end{proof}

A result of Bodlaender~\cite{B89} implies that the local treewidth of every planar graph is at most $r \mapsto 3r$.
Therefore,~\Cref{thm:local-tw} yields the following statement, which immediately implies~\Cref{th:planar-bd-radius}.

\begin{corollary}
    For every planar graph $G$ and for every positive integer $r$, the hypergraph of balls of radius at most $r$ in $G$ has a proper array sample compression scheme of size $24r+8$. 
\end{corollary}

\subsection{Bounded degeneracy}\label{sec:bd-degeneracy}

We conclude this section with a simple construction of sample compression schemes for the hypergraphs of closed neighborhoods in graphs of bounded degeneracy.

A graph $G$ has \defin{degeneracy} at most $t$ if there exists a total ordering $\prec$ on $V(G)$ such that for every $v \in V(G)$, the size of the set $\{u \in N_G(v) \mid u \prec v\}$ is at most $t$.

Before we prove~\Cref{thm:degeneracy}, let us delve more into the context.
Recall the family of graphs $G_n$ for positive integers $n$ defined in the introduction of this paper.
Note that the graph $G_n'$ obtained by subdividing each edge of $G_n$ once has degeneracy $2$, and the VC-dimension of its hypergraph of balls of radius (at most) $2$ is at least $n$.
It follows that there is no constant $c$ such that the hypergraphs of balls of radius (at most) $2$ in graphs of degeneracy $2$ admit sample compression schemes of size $c$.

\thmdeg*

One can check that $G_n$ has degeneracy $n-1$. 
Thus, as usual, using the result of Floyd and Warmuth~\cite{FW95}, we obtain that the linear bound in~\Cref{thm:degeneracy} is optimal.
Finally, note that the hypergraphs of closed neighborhoods in a graph are self-dual, hence, the general result of Moran
and Yehudayoff~\cite{Moran2016} yields that the hypergraphs of closed neighborhoods in graphs of degeneracy at most $t$ admit sample compression schemes of size $\Oh(t^2 \log t)$.

\begin{proof}[Proof of \Cref{thm:degeneracy}]
    Let $t$ be a positive integer and let $G$ be a graph of degeneracy at most $t$. Fix a total ordering $\prec$ on $V(G)$ which witnesses that $G$ has degeneracy at most $t$.
    Let $H$ be the hypergraph of closed neighborhoods of $G$.
    %We will use implicitly in the construction the fact that for every $x \in V(G)$, the size of the set $\{u \in N_G(x) \cup \{x\} \mid u \preceq x\}$ is at most $t+1$.
    
    \header{Compression.}
            Let $(X^+,X^-)$ be a sample of $H$ and fix $c\in V(G)$ such that $B_G(c,1) \cap (X^+ \cup X^-) = X^+$.
            Consider the following two cases.
            \begin{enumorig}[label={Case \arabic*:}, ref={\arabic*}, left=20pt]
                \item there exists $x \in X^+$ with $c \preceq x$. \label{case:deg1}
                \item for all $x \in X^+$, we have $x \prec c$. \label{case:deg2}
            \end{enumorig}
            Let $b \in \{0,1\}$ be a bitstring encoding which case out of the two holds.
            Assume first that Case~\ref{case:deg1} holds.
            We fix $x \in X^+$ with $c \preceq x$.
            Next, let $j \in [t+1]$ be such that $c$ has the ordinal number $j$ in the set $\{u \in N_G(x) \cup \{x\} \mid u \preceq x\}$ sorted according to $\prec$.
            Let $b' \in \{0,1\}^{\lceil \log(t+1)\rceil}$ be a bitstring encoding the number $j$.
            When Case~\ref{case:deg2} holds, let $b' \in \{0,1\}^{\lceil \log(t+1)\rceil}$ be any bitstring.
            Depending on the two cases, we define a subsample of $(X^+,X^-)$:
            %-- suppose that case $i \in [2]$ holds:
            \[(Y^+,Y^-) = \begin{cases}
                (\{x\},\emptyset) & \text{if Case~\ref{case:deg1} holds,}\\
                (X^+,\emptyset)   & \text{if Case~\ref{case:deg2} holds.}
            \end{cases}\]
            Note that in Case~\ref{case:deg2}, the size of $X^+$ is at most $t$ as $X^+ \subset \{u \in N_G(c) \mid u \prec c\}$.
            Finally, set
            \[
                \kappa(X^+, X^-) = ((Y^+,Y^-), bb').
            \]   

        \header{Reconstruction.}
            Consider an input for the reconstructor $\rho$:
            \[((Y^+,Y^-), bb') \in \calS(H) \times \{0,1\}^{1+\lceil \log(t+1)\rceil}.
            \]
            We assume that the input is of the form encoded by $\kappa$; otherwise it does not matter what the reconstructor returns.
            First, $b$ tells us in which case we are.
            If Case~\ref{case:deg1} holds, let $x$ be the element of $Y^+$, let $j$ be the number encoded by $b'$, and let $c'$ be the vertex with the ordinal number $j$ in the set $\{u \in N_G(x) \cup \{x\} \mid u \preceq x\}$ ordered according to $\prec$.
            Then, $\rho$ returns $B_G(c',1)$.
            If Case~\ref{case:deg2} holds, $\rho$ simply returns $Y^+$.

    \header{Proof of correctness.}
            It follows from the construction that $(\kappa,\rho)$ is a sample compression scheme for $H$ of the correct size.
\end{proof}

\Cref{thm:degeneracy} can be generalized to balls of radius at most $r$ using \emph{weak coloring numbers}.
Namely, with exactly the same idea one can prove that for every integer $r$ and for every graph $G$ with $r$-th weak coloring number at most $t$, the hypergraph of balls of radius at most $r$ has a sample compression scheme of size $\Oh(t)$.
Since we did not find any interesting application of this result, we elected not to include it.

% =======================================================
\section{Open problems}\label{sec:open}
% =======================================================

Recall that the general result of Moran and Yehudayoff \cite{Moran2016} implies that for every $K_t$-minor-free graph $G$, the hypergraph of balls in $G$ has a sample compression scheme of size $\Oh(t^2\log t)$.
We improved the bound to almost linear in the case of graphs of treewidth at most $t$, see~\Cref{thm:tw-first}.
This yields a natural direction for further research.

\begin{question}
    Is it true that for every $K_t$-minor-free graph $G$, 
    the hypergraph of balls in $G$ admits a sample compression scheme of size $\widetilde{\Oh}(t)$? 
\end{question}

Even though we consider the bounds in~\Cref{thm:tw-first,thm:cw-first} to be  \emph{almost} tight, it would be interesting to know if the logarithmic factors are necessary.

\begin{question}\label{q:remove-log-factor}
    Does every hypergraph of balls in a graph of treewidth (resp.\ cliquewidth) at most $t$ admit a (proper) sample compression scheme of size $\Oh(t)$?
\end{question}

\Cref{thm:vc} states that the hypergraphs of balls in graphs admitting a vertex cover of size at most $t$ have proper sample compression schemes of size $\Oh(t)$.
Thus a natural step towards \Cref{q:remove-log-factor} is to consider the case of graphs of bounded treedepth.

In most cases, we manage to devise proper sample compression schemes, but this often requires extra work, and it is not clear whether this is always possible. In particular, we think that the following directions are the most promising.

\begin{question}
    Is there a constant $c$ such that for every planar graph $G$, the hypergraph of balls in $G$ admits a proper sample compression scheme of size $c$? 
\end{question}

\begin{question}\label{question:Kt}
    Is there a function $f\colon \mathbb{N}\to\mathbb{N}$ such that for every positive integer $t$ and every $K_t$-minor-free graph $G$, the hypergraph of balls in $G$ admits a proper sample compression scheme of size $f(t)$? 
\end{question}

\begin{question}\label{question:closed_N}
    Is there a function $f\colon \mathbb{N} \to \mathbb{N}$
    such that every hypergraph of closed neighborhoods in a graph of degeneracy $t$ admits a proper sample compression scheme of size $f(t)$?
\end{question}

% \begin{question}\label{question:Littlewood}
%     Is there a function $f \colon \mathbb{N} \to \mathbb{N}$
%     such that every hypergraph of Littlestone dimension $t$ admits a proper sample compression scheme of size $f(t)$?
% \end{question}

For all we know, the functions in Questions~\ref{question:Kt} and \ref{question:closed_N} could be linear in $t$.

% ----------------------------
\bibliographystyle{abbrv}
\bibliography{bibliography}

\appendix

\section{Cliquewidth and VC-dimension of the hypergraph of balls}

In this section, we show that the hypergraph of balls in a graph of NLC-width at most $t$ has VC-dimension at most $6t+2$.
To do so, we follow the proof of Bousquet and Thomassé~\cite{BT15} that the hypergraph of balls
of a graph of rankwidth $t$ has VC-dimension at most $3\cdot2^{t+1}+2$.
Similarly as Bousquet and Thomassé~\cite{BT15}, 
we bound the \defin{2VC-dimension} of such hypergraphs $H$,
which is the maximum size of a \defin{$2$-shattered set} in $H$,
that is a set $S \subseteq V(H)$ such that for all distinct $s_1,s_2 \in S$,
there exists $e \in E(H)$ such that $e \cap S = \{s_1,s_2\}$.

\begin{proposition}\label{prop:bound-vcdim}
    Let $t$ be a positive integer,
    and let $G$ be a graph of NLC-width at most $t$.
    The 2VC-dimension of the hypergraph of the balls in $G$ is at most $6t+2$.
\end{proposition}

\begin{proof}
    Fix an NLC decomposition $(T,Q,\alpha,\beta,R)$ of $G$ with $|Q| \leq t$.
    We assume that $T$ has at least two nodes, as otherwise $|V(G)| \leq 1$ and the result is trivial.
    Suppose for contradiction that there is a set $S \subseteq V(G)$ of size $6t+3$
    which is $2$-shattered by balls in $G$.
    Recall that $V(G)$ is exactly the set of leaves of $T$.

    The following claim is a consequence of \cite[Lemma~6]{BT15}.
    We provide a proof for completeness.

    \begin{claim*}
        There is an edge $xy$ in $T$ such that
        $|S \cap V(T_{x \mid y})|, |S \cap V(T_{y \mid x})| \geq |S|/3 = 2t+1$.
    \end{claim*}

    \begin{proofclaim}
        Orient the edges of $T$ such that for every arc $(x,y)$,
        we have $|S \cap V(T_{y \mid x})| \geq |S|/2$.
        Consider now a sink $x$ in this orientation.
        Note that $x$ is not a leaf of $T$ since $|S|/2 > 1$.
        In particular, $x \notin S$.
        Since $x$ has degree at most three, and because $x \not\in S$,
        there exists $y \in N_T(x)$ such that $|S \cap V(T_{y \mid x})| \geq |S|/3$.
        Because $x$ is a sink, 
        we also have $|S \cap V(T_{x \mid y})| \geq |S|/2 \geq |S|/3$,
        and so the edge $xy$ is as desired.
    \end{proofclaim}

    Fix an edge $xy$ as in the claim,
    and assume that $y$ is the parent of $x$ in $T$.
    Let $U_1 = V(T_{x \mid y})$, $U_2 = V(T_{y \mid x})$, $S_1 = S \cap U_1$, and $S_2 = S \cap U_2$.

    For each $(s_1,s_2) \in S_1 \times S_2$, fix a ball $B_G(c_{(s_1,s_2)}, r_{(s_1,s_2)})$
    such that $B_G(c_{(s_1,s_2)}, r_{(s_1,s_2)}) \cap S = \{s_1,s_2\}$.
    Consider the directed graph $D$ with vertex set $S_1 \cup S_2$
    and arcs all the pairs $(s_1,s_2) \in S_1 \times S_2$
    such that $c_{(s_1,s_2)} \in U_1$,
    and all the pairs $(s_2,s_1) \in S_2 \times S_1$ such that $c_{(s_1,s_2)} \in U_2$.
    Note that $D$ has exactly $|S_1|\cdot|S_2|$ arcs and $|S_1| + |S_2|$ vertices.
    Therefore, there is a vertex $a$ of $D$ with out-degree at least $\frac{|S_1|\cdot|S_2|}{|S_1|+|S_2|} > t$.
    Fix $t+1$ out-neighbors $b_1, \dots, b_{t+1}$ of $a$ in $D$.
    For every $q \in Q$, let
    \[V_q = \{u \in U_1 \mid \beta(\{x,u\})(\alpha(u)) = q\}.\]
    By~\Cref{lem:NLC-few-nbhd-classes}, $(V_q \mid q \in Q)$ is a $U_2$-twin partition of $U_1$.
    We now split the proof into two cases: $a \in S_1$ and $a \in S_2$.

    First assume $a \in S_1$,
        and so $b_1, \dots, b_{t+1} \in S_2$.
        For every $i, j \in [t+1]$,
        by \Cref{lemma:separator_with_few_types} applied to
        $(c_0,d_0) = (c_{(a,b_j)}, b_i)$,
        $(C,D) = (U_1,U_2)$, and $(V_q \mid q \in Q)$, we have
        \begin{equation}
            \dist_G(c_{(a,b_{j})}, b_i) = \min_{q \in Q} \left(\dist_G(c_{(a,b_{j})}, V_q) + \dist_{G[V_q \cup U_2]}(V_q, b_i)\right).
            \label{eq:2VC_cw_case1}
        \end{equation}
        For each $i \in [t]$, we fix $q(i) \in Q$ such that
        \[
            \dist_G(c_{(a,b_i)}, b_i) = \dist_G(c_{(a,b_{i})}, V_{q(i)}) + \dist_{G[V_{q(i)} \cup U_2]}(V_{q(i)}, b_i).
        \]
        By the pigeonhole principle and because $|Q| \leq t$,
        there exist distinct $i,j \in [t+1]$ such that $q(i) = q(j)$.
        We write $q = q(i) = q(j)$.
        Without loss of generality, $\dist_{G[V_q \cup U_2]}(V_q, b_i) \leq \dist_{G[V_q \cup U_2]}(V_q, b_j)$.
        Then, 
        \begin{align*}
            \dist_G(c_{(a,b_{j})}, b_i) 
            &\leq \dist_G(c_{(a,b_{j})}, V_q) + \dist_{G[U_2 \cup V_q]}(V_q, b_i)
                && \textrm{by \eqref{eq:2VC_cw_case1}}\\
            &\leq \dist_G(c_{(a,b_{j})}, V_q) + \dist_{G[U_2 \cup V_q]}(V_q, b_j) \\
            &= \dist_G(c_{(a,b_{j})}, b_j) 
                && \textrm{since $q=q(j)$}\\
            &\leq r_{(a,b_{j})}.
        \end{align*}
        Therefore, $b_i \in B_G(c_{(a,b_{j})}, r_{(a,b_{j})})$,
        which contradicts the fact that $B_G(c_{(a,b_{j})}, r_{(a,b_{j})}) \cap S = \{a,b_{j}\}$.

    Next, assume that $a \in S_2$,
        and so $b_1, \dots, b_{t+1} \in S_1$.
        For every $i, j \in [t+1]$,
        by \Cref{lemma:separator_with_few_types} applied to
        $(c_0,d_0) = (b_i, c_{(a,b_j)})$,
        $(C,D) = (U_1,U_2)$, and $(V_q \mid q \in Q)$, we have
        \begin{equation}\label{eq:2VC_cw_case2}
            \dist_G(b_i, c_{(a,b_{j})}) = \min_{q \in Q} \left(\dist_G(b_i, V_q) + \dist_{G[V_q \cup U_2]}(V_q,  c_{(a,b_{j})})\right).
        \end{equation} 
        For each $i \in [t]$, we fix $q(i) \in Q$ such that
        \[
            \dist_G(b_i, c_{(a,b_i)}) = \dist_G(b_i, V_q) + \dist_{G[V_q \cup V_2]}(V_q, c_{(a,b_i)}).
        \]
        By the pigeonhole principle and because $|Q| \leq t$,
        there exist distinct $i,j \in [t+1]$ such that $q(i) = q(j)$.
        We write $q = q(i) = q(j)$.
        Without loss of generality, $\dist_{G}(b_i, V_q) \leq \dist_{G}(b_j, V_q)$.
        But then,
        \begin{align*}
            \dist_G(b_i,c_{(a,b_{j})}) 
            &\leq \dist_G(b_i, V_q) + \dist_{G[V_q \cup U_2]}(V_q, c_{(a,b_{j})})
                && \textrm{by \eqref{eq:2VC_cw_case2}} \\
            &\leq \dist_G(b_j, V_q) + \dist_{G[V_q \cup U_2]}(V_q, c_{(a,b_{j})}) \\
            &= \dist_G(b_{j}, c_{(a,b_{j})})
                && \textrm{since $q=q(j)$} \\
            &\leq r_{(a,b_{j})}.
        \end{align*}
        Therefore, $b_i \in B_G(c_{(a,b_{j})}, r_{(a,b_{j})})$,
        which contradicts the fact that $B_G(c_{(a,b_{j})}, r_{(a,b_{j})}) \cap S = \{a,b_{j}\}$.
        This contradiction concludes the proof.
\end{proof}

\end{document}